\documentclass[11pt,a4paper]{article}
\usepackage{t1enc}
\usepackage[utf8]{inputenc}
\usepackage[english]{babel}
\normalfont
\usepackage{amsmath}
\usepackage{amssymb}
\usepackage{amsthm}
\usepackage{yfonts}

\usepackage{mathabx}
\usepackage{bbm}
\usepackage{bm}
\usepackage{wasysym}
\usepackage{mathrsfs}
\usepackage{pifont}
\usepackage{hyperref}
\usepackage{pgf}
\usepackage{graphicx}
\usepackage{chngcntr}
\counterwithout{figure}{section}

\newcommand{\be}[0]{\begin{equation}}
\newcommand{\ee}[0]{\end{equation}}

\numberwithin{equation}{section}

\theoremstyle{plain}% default
\newtheorem{theorem}{Theorem}[section]

\newtheorem{proposition}[theorem]{Proposition}

\theoremstyle{definition}

\begin{document}

\vspace*{-1cm}
\thispagestyle{empty}
\vspace*{1.5cm}

\begin{center}
{\Large 
{\bf The global anomalies of (2,0) superconformal field theories in six dimensions}}
\vspace{2.0cm}

{\large Samuel Monnier}
\vspace*{0.5cm}

Institut für Mathematik, %Mathematisch-naturwissenschaftliche Fakultät, 
Universität Zürich,\\
Winterthurerstrasse 190, 8057 Zürich, Switzerland

\vspace*{1cm}

{\bf Abstract}
\end{center}

We compute the global gauge and gravitational anomalies of the A-type (2,0) superconformal quantum field theories in six dimensions, and conjecture a formula valid for the D- and E-type theories. We show that the anomaly contains terms that do not contribute to the local anomaly but that are crucial for the consistency of the global anomaly. A side result is an intuitive picture for the appearance of Hopf-Wess-Zumino terms on the Coulomb branch of the (2,0) theories.

\newpage

\tableofcontents

\section{Introduction and summary}

Global gravitational anomalies \cite{Witten:1985xe} are anomalous phases picked by the partition function of quantum field theories under large diffeomorphisms of spacetime. Just as for local anomalies \cite{AlvarezGaume:1983ig}, their cancellation is required in quantum field theories arising as low energy effective descriptions of quantum theories of gravity, providing constraints on the latter. In non-gravitational theories, however, global anomalies need not vanish.

The aim of this paper is to compute the global gravitational anomalies of the 6-dimensional conformal field theories with (2,0) supersymmetry \cite{Witten:1995zh, Strominger:1995ac}, henceforth referred to as (2,0) theories. There are two main motivations for this computation, that will be presented in turn.

As we will explain in Section \ref{SecRemAnom}, the global anomaly of a $d$-dimensional quantum field theory $\mathfrak{F}$ is captured by an $\mathbb{R}/\mathbb{Z}$-valued geometric invariant ${\rm An}_\mathfrak{F}$ of $d+1$-dimensional manifolds. A large class of such invariants are Chern-Simons invariants, whose value on a $d+1$-dimensional manifold $U$ is given by the integral of a characteristic form of degree $d+2$ over a $d+2$-dimensional manifold $W$ bounded $U$. The knowledge of the local anomaly essentially amounts to the knowledge of a characteristic form $I$ in dimension $d+2$, and in simple cases, such as complex chiral fermions, ${\rm An}_\mathfrak{F}(U)$ is indeed simply given by the Chern-Simons invariant of $I$. However, such a formula can be consistent only when $I$ yields an integer whenever integrated over a closed manifold $W$. Indeed, this ensures that ${\rm An}_\mathfrak{F}(U)$ is well-defined modulo $\mathbb{Z}$.

The local anomaly of (2,0) theories has been computed in \cite{Harvey:1998bx} for theories in the A-series, in \cite{Yi:2001bz} for the D-series and a general formula, also valid for the E-series, has been conjectured in \cite{Intriligator:2000eq}. Given these expressions, it is easy to check that the corresponding degree 8 characteristic form $I$ does not integrate to an integer on closed $8$-dimensional manifolds (see equation \eqref{EqLocAn20Theory}). This shows that the Chern-Simons invariant of $I$ does not exist, and it is therefore an interesting task to determine the geometric invariant computing the anomaly of the (2,0) theory. We will show that the latter can be seen as the sum of the would-be Chern-Simons invariant of $I$ and an extra term that does not contribute to the local anomaly. While ill-defined separately, these two terms combine into a well-defined invariant of 7-dimensional manifolds.

The second motivation for the study of the global anomaly of (2,0) theories comes from the fact that they generate an impressive collection of supersymmetric theories in lower dimensions upon reduction. When reduced on a 4-manifold $X$, the (2,0) theory yields a 2-dimensional quantum field theory that can inherit a global gravitational anomaly, translating into a failure of modular invariance. The knowledge of the global anomaly of the (2,0) theory on generic 6-dimensional manifolds allows us in principle to compute the failure of modular invariance in the 2-dimensional theory in terms of the geometry and topology of $X$.

When reduced on a Riemann surface, the (2,0) theory yields a 4-dimensional supersymmetric theory. The latter admits an S-duality group given by the mapping class group of the Riemann surface \cite{Verlinde:1995mz, Witten:1997sc, Gaiotto:2009we}. The fact that the 6-dimensional theory has a global gravitational anomaly translates into the fact that the S-duality transformation of the 4-dimensional partition function is anomalous \cite{Vafa:1994tf, Labastida:1998sk}. Again, the knowledge of the 6-dimensional global gravitational anomaly allows us in principle to compute the anomalous transformation of the 4-dimensional theories under S-duality.

We will not venture into this interesting research program in the present paper, but only keep it in mind as a strong motivation for the derivation of a general anomaly formula for the (2,0) theory.

We can carry out rigorous computation of the global anomaly only for A-type theories. We use the fact that the latter can be realized on a stack of M5-branes in M-theory \cite{Strominger:1995ac}. In particular, there is a limit in which a set of $n$ parallel non-intersecting M5-branes flows to the $A_{n-1}$ (2,0) theory at a generic point of its Coulomb branch, together with a free tensor multiplet corresponding to the center of mass of the brane system. We showed recently in \cite{Monnierb} that the global anomaly of non-intersecting M5-branes vanishes, as is expected from the consistency of M-theory. In the present paper, we use this fact to derive the global anomaly of the (2,0) theory, in the same spirit as the derivation of the local anomaly in \cite{Harvey:1998bx}. To do so, we consider the M5-brane system above and pick a tubular neighborhood containing it. As we know that anomalies cancel in an M-theory spacetime including (non-intersecting) M5-branes, the anomaly of M-theory in the tubular neighborhood is due entirely to the presence of the boundary, and can essentially be computed by evaluating the M-theory Chern-Simons term on the boundary. One then obtains the anomaly of the (2,0) theory by subtracting the anomaly of the center of mass, which can be deduced from recent results about the global anomaly of the self-dual field \cite{Monnier2011a, Monniera}. One can then check explicitly that the geometric invariant obtained is well-defined, in the sense discussed above.

There is an essentially unique way of expressing the geometric invariant of the $A_n$ (2,0) theory in terms of Lie algebra data, and this provides a natural formula for the anomaly of the other (2,0) theories, which is automatically compatible with the exceptional isomorphisms between members of the A-D-E series. We check that the corresponding geometric invariant is well-defined as well for Lie algebras in the D and E series. A derivation of this formula in the $D_n$ case should be possible using the realization of the latter by $n$ M5-branes on a $\mathbb{R}^5/\mathbb{Z}_2$ orbifold. In this paper, we only point out that the anomaly of the $\mathbb{R}^5/\mathbb{Z}_2$ orbifold is not understood globally. Just like for the (2,0) theory, the Chern-Simons term obtained from the index density describing the local anomaly is ill-defined. In this case, however, we do not know how to compute the correct global anomaly.

In Section \ref{SecOrigHWZTerms}, we also present a simple picture for the appearance of the Hopf-Wess-Zumino terms present on the Coulomb branch of the (2,0) theory. Those terms can be thought of as the topological modes of the C-field living between the M5-branes, which have to persist when we scale distance between the M5-branes to zero in order to obtain the (2,0) theory.

Another interesting point is that the anomaly formula we derive suggests that more data is needed to define the (2,0) theory that was previously expected. In addition to a simply laced Lie algebra, a smooth oriented 6-manifold $M$, a rank 5 R-symmetry bundle $\mathscr{N}$ over $M$ and a spin structure on $TM \oplus \mathscr{N}$, we seem to need a global angular differential cohomology class on $\mathscr{N}$. This is a differential cohomology class on the 4-sphere bundle $\tilde{M}$ associated to $\mathscr{N}$, restricting on each fiber to a normalized top differential cohomology class on $\tilde{M}$. In the M-theory realization of the A-type theories, a choice of global angular differential cohomology class is required in order to perform the decoupling of the center-of-mass tensor multiplet. We should mention that when the fourth Stiefel-Whitney class of $\mathscr{N}$ vanishes, a canonical choice is available.

A conceptual way to think of anomalies is in terms of a field theory (in the mathematical sense of the term) in one dimension higher \cite{Freed:2014iua}. The geometric invariant computed in this paper is the partition function of this anomaly field theory. Other aspects of the anomaly field theory will be explored elsewhere \cite{Monnierc}. We should also mention that a discussion of the relation between the quantum field theory on a stack of M5-branes and a non-abelian Chern-Simons 7-dimensional theory appeared in \cite{Fiorenza:2012tb}.

We add two remarks to clarify the assumptions made in this paper and the caveats of the derivation. \footnote{We thank the referee for raising this point.} First, the anomaly cancellation check of \cite{Monnierb} was not quite complete, as it was assumed that all 7-dimensional manifolds $U$ involved in anomaly computations are bounded by 8-dimensional manifolds $W$. It was shown in \cite{Monnierb} that the possible obstruction, given by a certain cobordism group, is at most torsion. If the cobordism group turns out not to vanish, then the check in \cite{Monnierb} is incomplete and it is in principle possible that M-theory backgrounds containing certain configurations of M5-branes are anomalous under certain combinations of large diffeomorphisms and C-field gauge transformations. In this paper, we make the likely assumption that no such anomalies exist. (Their existence would imply a fundamental inconsistency of M-theory).

Second, to keep the derivation simple, we assume in this paper that the cobordism group vanishes, therefore that every $U$ is bounded by a $W$. This allows us to compute in Section \ref{SecEvCSTerm} the anomaly inflow using differential forms on $W$. As will be shown in \cite{Monnierc}, we are not losing any information from this assumption, because the anomaly inflow computation can be carried out on $U$, using the corresponding differential cocycles, and it yields the same result.

The paper is organized as follows. Section \ref{SecRemAnom} presents the relation between global anomalies of $d$-dimensional quantum field theories and geometric invariants of $d+1$-dimensional manifolds. We also review the known local anomalies of the (2,0) theories and explain why the associated Chern-Simons invariants are ill-defined. In Section \ref{SecGeomM5-branes}, we present aspects of the geometry of M5-branes necessary for our computation of the global anomaly. The derivation of the global anomaly of the A-type (2,0) theories is found in Section \ref{SecGlobAnAn}. We show that the anomaly formula determines a well-defined geometric invariant of 7-manifolds and comment on the appearance of conformal blocks and on the Hopf-Wess-Zumino terms present on the Coulomb branch of (2,0) theories. Section \ref{SecGenGlobAnForm} presents the general anomaly formula, conjecturally also valid for the D- and E-type theories, as well as a proof that the associated geometric invariants are well-defined.

\section{Some remarks about anomalies}

\label{SecRemAnom}

The aim of this section is to explain informally how the global anomaly of a $d$-dimensional quantum field theory can be described by a geometric invariant of $d+1$-dimensional manifolds. In Section \ref{SecAnomCob}, we introduce the anomaly line bundle and explain that its holonomies and transition functions can be computed by evaluating a geometric invariant on mapping tori and twisted doubles, respectively. In Section \ref{SecExamAnom}, we give some examples of anomalous theories and their geometric invariants. We introduce in Section \ref{SecIncons20TheoryAn} the local anomaly of the (2,0) theory and deduce a natural guess for its global anomaly. We explain why this naive guess cannot be correct, providing a motivation for the more careful derivation in the following sections.

\subsection{Global anomalies and cobordisms}

\label{SecAnomCob}

A global symmetry of a field theory on a $d$-dimensional manifold $M$ is associated to a current $J$. The latter can be sourced by a background field $A$, which belongs to an infinite-dimensional space of background fields $\mathcal{B}$. Two common examples of such symmetries are a global internal symmetry, described by a pointwise action of a Lie group $G$ on the fields of the theory, and the isometry group of spacetime, acting by pullback on the fields. The associated currents are the symmetry current and the energy-momentum tensor. The corresponding background fields in these two examples are a non-dynamical gauge field coupling to the current, and a (Riemannian or Lorentzian) metric on $M$. 

We can also consider the local transformations associated to the global symmetry. In our first examples, such local transformations are generated by the action on the fields of a section $g$ of a $G$-bundle over $M$. In the second example, the local transformations are the diffeomorphisms of $M$, or a subset of those, if some structure necessary for the definition of the field theory needs to be preserved. While a local transformation does not leave the action invariant, its effect can be compensated by a corresponding transformation on the background fields. In the first example, this is achieved by changing the background gauge field by the gauge transformation associated to $g$. In the second example, this is achieved by pulling back the metric of $M$ via the diffeomorphism.

In the quantum theory, we say that the global symmetry suffers from an \emph{anomaly} if the quantum theory turns out not to be invariant under the combined action of the local transformations on the fields and on the background fields. More precisely, we can see the partition function of the quantum field theory (as well as the associated correlation functions) as functions over the space of background fields $\mathcal{B}$. An anomaly is present if these functions are not invariant under the action of the group $\mathcal{G}$ of local transformations on $\mathcal{B}$. For unitary theories, the lack of invariance of the partition function $Z$ is only by a phase. Our aim in the present paper is to give a formula for these phases in the case of the 6-dimensional superconformal theories with (2,0) supersymmetries, when the local transformations are diffeomorphisms of the 6-dimensional spacetime.

A fruitful point of view on anomalies is the following. If $Z$ is not invariant under $\mathcal{G}$, it cannot define a function on the quotient $\mathcal{B}/\mathcal{G}$, seen as the space of gauge invariant background field data. However, $Z$ does define a section of a unitary $\mathcal{G}$-equivariant line bundle on $\mathcal{B}$. For all practical purposes, a $\mathcal{G}$-equivariant line bundle on $\mathcal{B}$ can be taken as the definition of a line bundle over $\mathcal{B}/\mathcal{G}$, valid even when the quotient is singular. Therefore, instead of defining a function over $\mathcal{B}/\mathcal{G}$, in general $Z$ defines a section of a unitary line bundle $\mathscr{L}$ over $\mathcal{B}/\mathcal{G}$.

From now on, in order to have a unified treatment, we include in the space of background field $\mathcal{B}$ all the data required to define our quantum field theory. In particular, a point of $\mathcal{B}$ specifies the $d$-dimensional spacetime $M$. $\mathcal{G}$ is then not exactly a group, but a groupoid obtained by the union of the groups of local transformations for each $M$, acting each on the respective component of $\mathcal{B}$. In this more general setting, the partition function still defines the section of a line bundle $\mathscr{L}$ over $\mathcal{B}/\mathcal{G}$.

How can we describe unitary line bundles and their sections over $\mathcal{B}/\mathcal{G}$? One way to do so is to pick some $\mathbb{R}/\mathbb{Z}$-valued geometric invariant of manifolds with boundary of dimension $d+1$. (Recall that the spacetime $M$ has dimension $d$.) We will write ${\rm An}_{\mathfrak{F}}$ for the geometric invariant describing the anomaly bundle of the quantum field theory $\mathfrak{F}$. By \emph{geometric invariant}, we mean a functional that depends on certain geometric or topological data on the $d+1$-dimensional manifold $U$, which after restriction to $\partial U$ defines a unique point in $\mathcal{B}$. The only requirement we put on ${\rm An}_{\mathfrak{F}}$ is that it is consistent with the gluing of manifolds along their boundaries. If $U_1$ has a boundary component $M$ and $U_2$ has a boundary component $\bar{M}$ ($M$ with the opposite orientation) such that the extra structure glues smoothly into a manifold $U_1 \cup_M U_2$, then we require that
\be
\label{EqFunctRelGeomInv}
{\rm An}_{\mathfrak{F}}(U_1) + {\rm An}_{\mathfrak{F}}(U_2) = {\rm An}_{\mathfrak{F}}(U_1 \cup_M U_2) \;.
\ee

In more abstract terms, we need to find a cobordism category $\mathfrak{C}$ whose objects are the elements of $\mathcal{B}$, i.e. $d$-dimensional manifolds endowed with all the structures we need to define our quantum field theory. ${\rm An}_{\mathfrak{F}}$ is then a functor from $\mathfrak{C}$ to the category whose only object is the complex line $\mathbb{C}$ and whose morphisms from $\mathbb{C}$ to itself are labeled by $U(1)$, identified with $\mathbb{R}/\mathbb{Z}$ via exponentiation.

The geometric invariant ${\rm An}_{\mathfrak{F}}$ then defines a unitary line bundle $\mathscr{L}$ with connection over $\mathcal{B}/\mathcal{G}$. For instance, a cobordism $U_b$ between the empty manifold and $b \in \mathcal{B}$ can be seen as defining the value at $b$ of (the pull-back of) a section $s$ of $\mathscr{L}$. Indeed, $b \in \mathcal{B}$ defines a manifold $M$ together with background fields, and there is a subset $\mathcal{B}_U \in \mathcal{B}$ consisting of the data that can be extended to $U$. The pull-back of $\pi^\ast(\mathscr{L})$ to $\mathcal{B}_U$ is trivial and we define $\pi^\ast(s)(b) = {\rm An}_{\mathfrak{F}}(U)$. As $b$ moves in $\mathcal{B}_U$, we obtain a function over $\mathcal{B}_U$, which is the pull-back of a section $s$ of $\mathscr{L}$. 

An element $g \in \mathcal{G}$ acts on $\mathcal{B}$ and induces a change of trivialization in $\pi^\ast(\mathscr{L})$. We can compute the phase of this change of trivialization by comparing the value of the pull-back of a given section $s$ at $b$ and at $g.b$. We know that $\pi^\ast(s)(b) = {\rm An}_{\mathfrak{F}}(U_b)$. Consider now the twisted double $U_g$ of $U_b$. This is the manifold obtained by gluing $U_b$ to $\bar{U}_b$ ($U_b$ with the opposite orientation) through the transformation $g$. Then ${\rm An}_{\mathfrak{F}}(U_g)$ is the logarithm of the phase associated to the change of trivialization induced by $g$. A simple reasoning shows that the phase obtained is independent of the choice of manifold $U_b$, i.e. of the choice of section of $\mathscr{L}$, see Figure \ref{Fig2}.

% ---- DO NOT ERASE, FIGURE TO BE INCLUDED IN THE PDF VERSION ----
\begin{figure}[p]
  \vspace{-1cm}
  \centering
  \includegraphics[width=\textwidth]{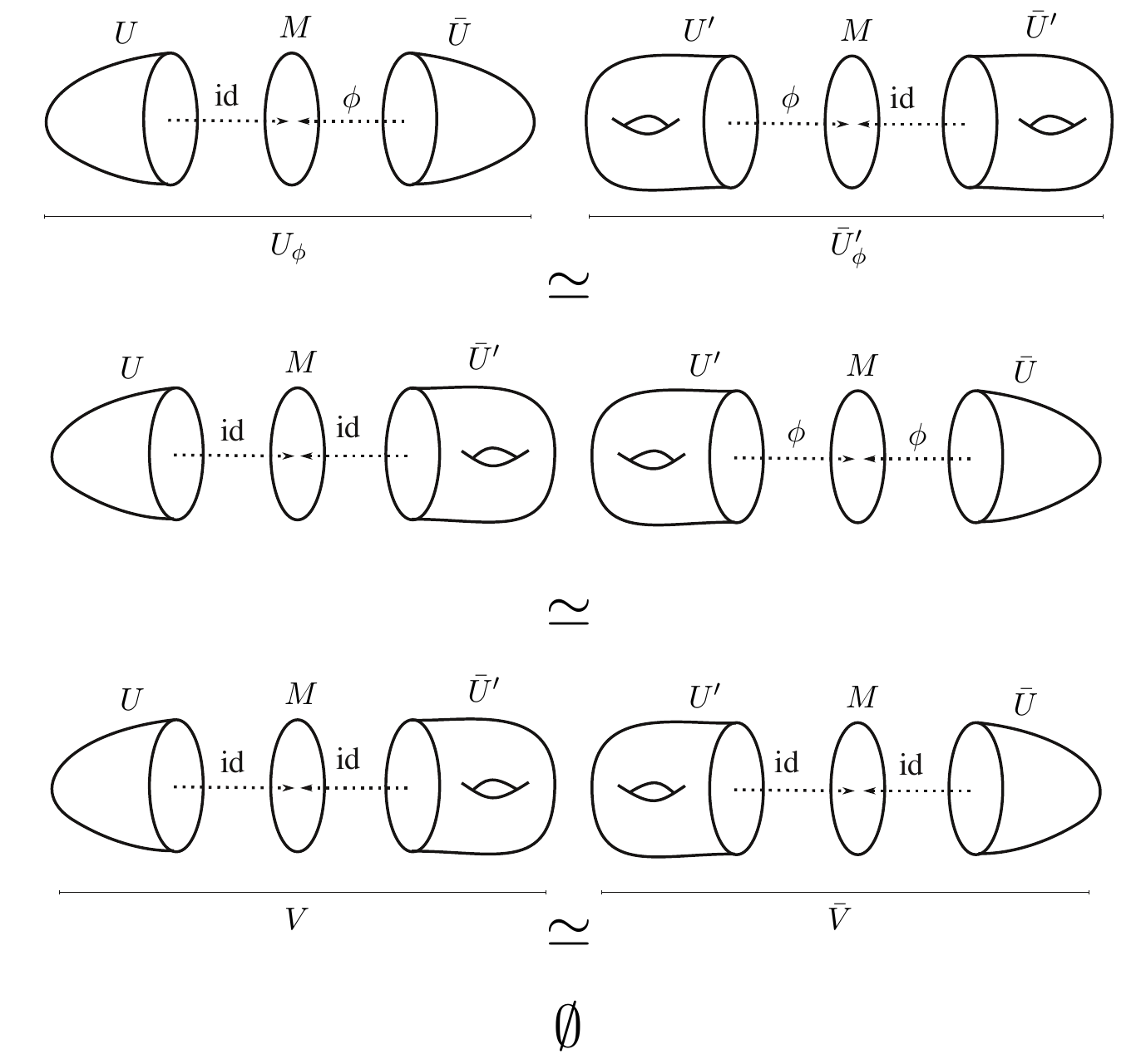}
  \caption{\emph{This figure illustrates the argument showing that the value of ${\rm An}_{\mathfrak{F}}$ on twisted doubles depends only on the gluing map $\phi$. We start by picking two manifolds $U$ and $U'$ bounded by $M$. On the top left, the twisted double $U_\phi$ is constructed by gluing two copies of $U$, one of them with its orientation reversed, with the help of the map $\phi$. On the top right, the same construction starting from $U'$, with the opposite orientation, yielding $\bar{U}'_\phi$. By rearranging the pieces, we obtain the second line. Then, noticing that the two twists cancel in the second gluing on the second line, we obtain on the third line $V = U \cup_{\rm id} \bar{U}'$ and $\bar{V}$. This pair of manifolds is bordant to the empty manifold, showing that ${\rm An}_{\mathfrak{F}}$ was zero all along and implying that ${\rm An}_{\mathfrak{F}}(U_\phi) = {\rm An}_{\mathfrak{F}}(U'_\phi)$. In terms of the line bundle $\mathscr{L}$, this translates into the fact that the transition functions do not depend on the sections used to compute them.}}
  \label{Fig2}
\end{figure}
% ---- DO NOT ERASE, FIGURE TO BE INCLUDED IN THE PDF VERSION ----

The parallel transport along a path $p$ in $\mathcal{B}$ is given by a cylindrical cobordism $U_{[0,1]} = M \times [0,1]$ between $p(0)$ and $p(1)$, in such a way that $M \times \{t\} = p(t)$. In particular, a loop $c \in \mathcal{B}$ determines a closed $d+1$ manifold $U_c$, the mapping torus associated to $c$. $\exp 2\pi i {\rm An}_{\mathfrak{F}}(U_c) \in U(1)$ is the holonomy of the connection on $\mathscr{L}$ along $c$. This explains the appearance of mapping tori in the computations of global anomalies \cite{Witten:1985xe, Witten:1985mj, Monnier2011a, Monnierb}. (We are here glossing over the fact that the path or loop in $\mathcal{B}$ might not define unambiguously the data needed to compute the geometric invariant on the cylinder or mapping torus. Those subtleties will play no role in what follows.)

Let us remark that the construction using twisted doubles reviewed above allows us to compute the anomalous phases picked by the partition function of a quantum field theory without computing the latter explicitly, provided we know the invariant ${\rm An}_{\mathfrak{F}}$.

\subsection{Examples}

\label{SecExamAnom}

Let us now turn to some examples. An important example is 3-dimensional Chern-Simons theory, in which the above is well-known \cite{Witten:1988hf, Labastida1989, Freed:1992vw}. The anomalous field theory is the 2-dimensional chiral WZW model and ${\rm An}_{\rm WZW}$ is the Chern-Simons functional. Depending on whether we are considering the quantum \cite{Witten:1988hf, Labastida1989} or classical \cite{Freed:1992vw} theory, we consider the gauge field as dynamical or include it with the background fields. The anomaly line bundle associated to a surface by the Chern-Simons term is the line bundle over the moduli space of flat connections of which the WZW conformal blocks are sections. This line bundle extends as well over the space of conformal structures of the surface.

Another example, treated in detail in \cite{Dai:1994kq}, is the modified eta invariant $\xi$ of a Dirac operator $D$ on an odd-dimensional manifold of dimension $d+1$. $\xi$ is related to the eta invariant by $\xi = \frac{1}{2} (\eta + h)$, where $h$ denotes dimension of the kernel of $D$. It was shown in \cite{Dai:1994kq} that when we take ${\rm An}_{D_+} = \xi$, then $\mathscr{L}$ is the inverse of the determinant bundle of the associated chiral Dirac operator $D_+$ in dimension $d$. ${\rm An}_{D_+}$ then computes the global anomaly of the complex chiral fermionic theory in dimension $d$ associated to the Dirac operator $D_+$. In particular, the holonomies of the anomaly connection are given by $\tau = \exp 2\pi i \xi$ evaluated on mapping tori, in a suitable adiabatic limit in which the size of the base of the mapping tori tends to infinity \cite{Witten:1985xe, MR861886}. One can also compute the actual phase picked by the chiral fermion partition function under a diffeomorphism or a gauge transformation by evaluating $\tau$ on a twisted double, as explained above.

The latter example has the following interesting property. Assume that a closed $d+1$-dimensional spin manifold $U$ is bounded by a $d+2$-dimensional manifold $W$ on which the spin structure of $U$ extends, as well as any other data required to define $D$. The invariant $\xi$ can be computed using the Atiyah-Patodi-Singer index theorem \cite{Atiyah1973}:
\be
\xi = -{\rm index}(D_W) + \int_W I_{D_W}
\ee
where $D_W$ is a Dirac operator on $W$ restricting to $D$ on $U$, and $I_{D_W} = \hat{A}(TW) {\rm ch}(E)$ is the index density of $D_W$. (We expressed $D_W$ as the ordinary Dirac operator on $W$ twisted by a vector bundle $E$). Note that we are reading this formula only modulo 1, so the first term on the right-hand side is irrelevant. 

$I_{D_W}$ is exactly the characteristic form in $d+2$ dimension used to compute the local anomaly of the chiral fermionic theory \cite{AlvarezGaume:1983ig}. It can be related to the curvature of the anomaly line bundle $\mathscr{L}$ as follows. Recall that the holonomies of the connection on $\mathscr{L}$ can be computed by evaluating $\tau$ on a mapping torus. Assume that we are interested in a small homotopically trivial loop $c$ in $\mathcal{B}$. Then the mapping torus $U_c = M \times S^1$ is trivial and we can take $W = M \times D^2$, where $D^2$ is a 2-dimensional disk. We find therefore that the holonomy around $c$ is given by the integral of $I_{D_W}$ over $M \times D^2$. But the holonomy around $c$ is also given by the integral of the curvature of $\mathscr{L}$ over $D^2$. As this is true for all loops $c$, we find that the curvature of $\mathscr{L}$ is given by the degree 2 component of $\int_M I_{D_W}$, where $I_{D_W}$ is seen as a differential form on $M \times \mathcal{B}$. 

We deduce that the local anomaly polynomial, of degree $d+2$, of a quantum field theory is directly related to the curvature of the anomaly line bundle via integration over spacetime. Of course, the local anomaly does not capture all the information about the anomaly of a quantum field theory: there exist line bundles with non-trivial flat connections. The set of holonomies of the connection captures all the information about the anomaly and is refered to as the \emph{global anomaly} \cite{Witten:1985xe}. Equivalently, the anomaly is fully captured by the geometric invariant ${\rm An}_{\mathfrak{F}}$, and this is the point of view that we will take in this paper.

\subsection{The case of the (2,0) theory}

\label{SecIncons20TheoryAn}

Let us now focus on the (2,0) theory in six dimensions. The local gravitational anomaly of the (2,0) theory of type $A_n$ was derived from M-theory in \cite{Harvey:1998bx}. This result was extended to the type $D_n$ case in \cite{Yi:2001bz} and a general formula also valid for the E-type theories was conjectured in \cite{Intriligator:2000eq}. The degree 8 local anomaly polynomial reads
\be
\label{EqLocAn20Theory}
I_\mathfrak{g} = r(\mathfrak{g}) J_8 - \frac{|\mathfrak{g}|{\rm h}_\mathfrak{g}}{24} p_2(\mathscr{N}_W) \;.
\ee
$r(\mathfrak{g})$, $|\mathfrak{g}|$ and ${\rm h}_\mathfrak{g}$ denote respectively the rank, the dimension and the dual Coxeter number of the simple and simply laced Lie algebra $\mathfrak{g}$. $J_8$, whose explicit expression will appear below, is the anomaly polynomial for a single tensor multiplet in six dimensions. $p_2(\mathscr{N}_W)$ is the second Pontryagin class of the rank 5 bundle $\mathscr{N}_W$ over $W$ obtained by extending the R-symmetry bundle  of the (2,0) theory on $M$.

Given our experience with chiral fermions, one may be optimistic and guess that the value of the geometric invariant ${\rm An}_{\mathfrak{g}}$ governing the global anomaly of the (2,0) theory, evaluated on a manifold $U$ bounded by $W$ is simply given by
\be
\label{EqGuessAn20Theory}
\frac{1}{2\pi i} \ln {\rm An}_{\mathfrak{g}}(U) = \int_W I_\mathfrak{g} \;, \quad {\rm mod} \; 1 \;.
\ee
The problem is that \eqref{EqGuessAn20Theory} is inconsistent. 
An $\mathbb{R}/\mathbb{Z}$-valued geometric invariant on a $d+1$-dimensional manifold $U$ defined by integrating certain characteristic form $I$ on a bounded manifold of dimension $d+2$ can be well-defined only if $\int_W I$ is an integer for any closed manifold $W$. This is manifestly not the case for \eqref{EqGuessAn20Theory}. For instance, as we will see, $J_8$ can be written $\frac{1}{8}L(TW) + \frac{1}{2}I_f$, where $L(TW)$ is the Hirzebruch L-genus and $I_f$ is an index density with $\int_W I_f \in 2 \mathbb{Z}$ for $W$ closed. On a closed manifold, we have $\int_W L(TW) = \sigma_W$, the signature of the 8-dimensional manifold $W$. If $r(\mathfrak{g}) \sigma_W$ is not a multiple of 8, and in general it has no reason to be so, then $\int_W L(TW)$ cannot define a geometric invariant on $U$. The second term in \eqref{EqLocAn20Theory} does not define an invariant of $U$ either. One can check explicitly that $|\mathfrak{g}|{\rm h}_\mathfrak{g}$ is a multiple of 6, but $\int_W p_2(\mathscr{N}_W)$ has no particular evenness property on a closed manifold. As the coefficients of the two terms do not vary proportionally when we change $\mathfrak{g}$, there is no hope that \eqref{EqGuessAn20Theory} can be well-defined.

The problem calls therefore for a more careful study. We will see that \eqref{EqGuessAn20Theory} holds after adding extra terms on the right hand side that do not contribute to the local anomaly and that make the geometric invariant \eqref{EqGuessAn20Theory}  well-defined. Our strategy will be to focus first on A-type theories, through their realizations on stacks of M5-branes. We will then find a straightforward generalization to the D- and E-type theories.

\section{The geometry of M5-branes}

\label{SecGeomM5-branes}

In this section, we review some facts about the geometry of M5-branes that will be useful in the derivation of the anomaly of the (2,0) theory. In Section \ref{SecDisM5br}, we review the properties of the embedding of the M5-brane worldvolume into spacetime. Some subtleties about the coupling of the M-theory C-field to the worldvolume theory of the M5-brane are reviewed in Section \ref{SecEffC-field} and we introduce some important notations for the rest of the paper. In Section \ref{SecStackM5}, we generalize the analysis to the case of a stack of M5-branes and in Section \ref{SecExt7-8-manifolds}, we extend these constructions to 7- and 8-dimensional manifolds, as required by anomaly computations. In Section \ref{SecAnCancDisM5}, we review the M-theory Chern-Simons term and its role in anomaly cancellation in the presence of M5-branes.

We assume that the reader is familiar with the use of (shifted) differential cocycles to model higher $p$-form abelian gauge fields. The original reference is \cite{hopkins-2005-70}. An introduction for physicists can be found in Section 2 of \cite{Freed:2006yc}. Our notations follow Section 2.1 of \cite{Monnierb}, which can be read as a quick reminder. Differential cocycles and cohomology classes are written with a caron $\check{}$. What we often call the \emph{field strength} of a differential cocycle is sometimes called the curvature in the literature. The reason for our terminology is obvious: when the differential cocycle models an abelian gauge field, its curvature coincides with the field strength of the gauge field.

\subsection{Non-intersecting M5-branes}

\label{SecDisM5br}

We consider the low energy limit of M-theory on an 11-dimensional smooth oriented spin manifold $Y$, in the limit of vanishing gravitational coupling. It consists in 11-dimensional supergravity, together with a Chern-Simons term involving an important higher derivative correction \cite{Witten:1996md}. We work in Euclidean signature, so we take $Y$ to be Riemannian. We will be considering gauge transformations and diffeomorphisms that are the identity outside of a compact subset $U$ of the spacetime $Y$. This implies that we can freely modify $Y$ outside this subset and take it to be compact, possibly adding sources outside $U$ in order to satisfy the Gauss law of the gauge fields. 

Inside $U$, we choose a smooth oriented 6-dimensional manifold $M$, and we wrap one M5-brane on each of its connected components. We write $\mathscr{N}$ for the normal bundle of $M$ in $Y$. Our assumptions that $Y$ is oriented and spin and that $M$ is oriented imply that
\be
\label{EqRelSWClassesTMN}
w_1(TM) = w_1(\mathscr{N}) = 0 \;, \quad w_2(TM) + w_2(\mathscr{N}) = 0 \;, \quad w_5(\mathscr{N}) = 0 \;.
\ee
The last equality is not obvious and its proof can be found in Appendix A of \cite{Monnierb}. It should be also emphasized that it ceases to be automatically true once we extend these constructions from the 6-dimensional manifold $M$ to an 8-dimensional manifold $W$. In this case it will assumed.

We pick a tubular neighborhood $N$ of $M$ of radius $\delta$, which will eventually be taken to zero, and we write $\tilde{M}$ for its boundary. $\tilde{M}$ is a 4-sphere bundle over $M$, and we write $\pi$ for the bundle map $\tilde{M} \rightarrow M$. We have
\be
T\tilde{M} \oplus \mathbb{R}_{\tilde{M}} \simeq \pi^\ast(TM \oplus \mathscr{N}) \;, 
\ee
with $\mathbb{R}_{\tilde{M}}$ a trivial line bundle over $\tilde{M}$. This implies that for any stable characteristic class $c$, such as the Pontryagin or Stiefel-Whitney classes, we have
\be
\label{EqRelCharClassMtMN}
c(T\tilde{M}) = \pi^\ast(c(TM \oplus \mathscr{N})) \;.
\ee

\subsection{The effective C-field}

\label{SecEffC-field}

It is well known that the quantization of the fluxes of the M-theory C-field are shifted: they are integral or half-integral depending on the parity of the periods of $w_4(TY)$ \cite{Witten:1996md}. The precise way of encoding this statement is to see the C-field as an element $\check{C}$ of a group of shifted differential cocycles, written $\check{Z}_{\lambda}$ in Section 2.1 of \cite{Monnierb}. The shift $\lambda$ is a half-integer-valued cocycle such that $2\lambda$ is a lift of $w_4(TY)$ as an integral cocycle (i.e. a period of $2\lambda$ on a 4-cycle is even or odd depending on whether the period of $w_4(TY)$ is 0 or 1). From now on, we will refer to this shift simply as \emph{a shift by} $w_4(TY)$. We will similarly encounter later differential cocycles shifted by the degree 4 Wu class of $M$, i.e. by $w_4(TM) + (w_2(TM))^2$. 

The M-theory C-field sources the self-dual two-form gauge field on the worldvolume of the M5-brane. However, it is not trivial to restrict the C-field to the M5-brane worldvolume. Indeed, the M5-brane itself sources the C-field in the bulk, which means that the integral of the C-field field strength $G$ on any 4-sphere linking $M$ is equal to 1. This implies that $G$ diverges near $M$. If the normal bundle $\mathscr{N}$ is trivial, a trivialization defines longitudinal and normal components of elements of $T^\ast M$. The divergent part of the four-form $G$ is purely normal, and one can restrict the longitudinal component to $M$. However, this strategy does not work if the normal bundle is non-trivial.

In Section 2.3 of \cite{Monnierb}, we explained how to define in the general case the effective C-field on the worldvolume. In terms of differential cocycles, the restriction reads
\be
\label{EqDefRestC-field}
\check{C}_M = \frac{1}{2}\pi_\ast(\check{C}_{\tilde{M}} \cup \check{C}_{\tilde{M}}) \;.
\ee
Here $\pi_\ast$ is the pushforward map on differential cocycles associated to the fiber bundle $\tilde{M} \stackrel{\pi}{\rightarrow} M$, $\check{C}_{\tilde{M}}$ is the (non-singular) restriction to $\tilde{M}$ of the C-field on $Y$, and $\cup$ is the cup product on differential cocycles. Let us remark that the factor $\frac{1}{2}$ in \eqref{EqDefRestC-field} makes it not obvious that the differential cohomology class of $\check{C}_M $ depends only on the differential cohomology class of $\check{C}_{\tilde{M}}$, i.e. that \eqref{EqDefRestC-field} is gauge invariant. This can be shown by performing an explicit gauge transformation on $\check{C}_{\tilde{M}}$ in \eqref{EqDefRestC-field} and noticing that the factor $\frac{1}{2}$ do not appear in the variation of $\check{C}_M$ \cite{Monnierb}.

One can show that this definition reduces to the intuitive one sketched above when $\mathscr{N}$ is trivial. We also showed in \cite{Monnierb} that it passes a highly non-trivial consistency test: $\check{C}_M$ is a differential cocycle on $M$ shifted by the degree 4 Wu class of $M$, which is exactly what is required to define consistently the coupling to the worldvolume self-dual field \cite{Witten:1996hc}. (To be precise, the degree 4 Wu class of $M$ always vanishes, for dimensional reasons. We will however momentarily extend these constructions to manifolds of dimension 8, whose degree 4 Wu classes can be non-trivial.)

For explicit computations, it will be useful to choose an unshifted differential cocycle $\check{a}_{\tilde{M}}$, whose field strength $f_{\tilde{M}}$ integrates to $1$ over the 4-sphere fibers to $\tilde{M}$. We will refer to $\check{a}_{\tilde{M}}$ in the following as a \emph{global angular differential cocycle}. $\check{C}_{\tilde{M}}$ and its field strength $G_{\tilde{M}}$ can then be written
\be
\label{EqParCFieldTMSinM5}
\check{C}_{\tilde{M}} = \check{a}_{\tilde{M}} + \pi^\ast(\check{A}_M) \;, \quad G_{\tilde{M}} = f_{\tilde{M}} + \pi^\ast(F_M) \;,  
\ee
for some differential cocycle $\check{A}_M$ shifted by $w_4(TM \oplus \mathscr{N})$, with field strength $F_M$. The coefficient of $\check{a}_{\tilde{M}}$ in \eqref{EqParCFieldTMSinM5} is 1 because the M5-brane supported on $M$ sources one unit of flux of the C-field. The effective C-field \eqref{EqDefRestC-field} and its field strength $G_M$ then read
\be
\check{C}_M = \check{b}_M + \check{A}_M \;, \quad G_M = h_M + F_M \;,
\ee
where we defined
\be
\label{EqDefb}
\check{b}_M = \frac{1}{2} \pi_\ast(\check{a}_{\tilde{M}} \cup \check{a}_{\tilde{M}})
\ee
and wrote $h_M$ for the field strength of $\check{b}_M$. The differential cocycle $\check{b}_M$ gives rise to a well-defined differential cohomology class for the same reason as $\check{C}_M$ does, see \eqref{EqDefRestC-field}. Results of \cite{Witten:1999vg} show that it is shifted by $w_4(\mathscr{N})$. The differential cocycles $\check{a}_{\tilde{M}}$ and $\check{b}_{M}$ will play an important role in what follows.

\subsection{Stacks of M5-branes}

\label{SecStackM5}

We point out here the differences arising when $M$ supports a stack of $n$ M5-branes, rather than a single one. The flux through the fibers of $\tilde{M}$ is now $n$ units. Using \eqref{EqParCFieldTMSinM5}, we can parameterize the C-field on $M$ as follows:
\be
\label{EqParCFieldTMStaM5}
\check{C}_{\tilde{M}} = n \check{a}_{\tilde{M}} + \pi^\ast(\check{A}_M) \;, \quad G_{\tilde{M}} = nf_{\tilde{M}} + \pi^\ast(F_M) \;.  
\ee
$\check{A}_M$ is as before a differential cocycle shifted by $w_4(TM \oplus \mathscr{N}_M)$. Under changes of the parameterization \eqref{EqParCFieldTMStaM5}, we have
\be
\label{EqReparC-fieldtM1}
\check{a}_{\tilde{M}} \rightarrow \check{a}_{\tilde{M}} + \pi^\ast(\check{B}_M) \;, \quad \check{A}_M \rightarrow \check{A}_M - n \check{B}_M 
\ee
for $\check{B}_M$ an unshifted differential cocycle on $M$. We can also define
\be
\label{EqDefEffCn}
\check{C}_{M,n} := n\check{b}_M + \check{A}_M \;, \quad G_{M,n} := nh_M + F_M \;,
\ee
which is invariant under \eqref{EqReparC-fieldtM1}. We can also write $\check{C}_{M,n} = \frac{1}{2n}\pi_\ast(\check{C}_{\tilde{M}} \cup \check{C}_{\tilde{M}})$. Depending on whether $n$ is even or odd, $\check{C}_{M,n}$ is shifted by $w_4(TM \oplus \mathscr{N})$ or by the degree 4 Wu class of $M$. The differential cocycle
\be
\label{EqDefEffC1}
\check{C}_M := \check{b}_M + \check{A}_M
\ee
is shifted by the Wu class of $M$ and will play an important role in what follows. Remark that $\check{C}_M$ depends on a choice of parameterization \eqref{EqParCFieldTMStaM5}.

Simplifications occur when $w_4(\mathscr{N}) = 0$. Indeed, consider the vertical tangent bundle $T_V \tilde{M}$. Remark that its Euler class $e(T_V \tilde{M})$ integrates to 2 over the 4-sphere fibers of $M$, because the Euler number of a 4-sphere is 2. Modulo 2, we have 
\be
e(T_V \tilde{M}) = w_4(T_V \tilde{M}) = \pi^\ast(w_4(\mathscr{N})) \;.
\ee
Therefore, if $w_4(\mathscr{N}) = 0$, $e(T_V \tilde{M})$ can be divided by 2. \cite{Bott1998} shows that $\pi_\ast\left( \frac{1}{2} e(T_V \tilde{M}) \cup \frac{1}{2}e(T_V \tilde{M})\right)$ is at most torsion. The above holds for the differential refinement $\check{e}(T_V \tilde{M})$ obtained from the metric on $\tilde{M}$. We can therefore take $\check{a}_{\tilde{M}} = \frac{1}{2} \check{e}(T_V M) + \pi^\ast(\check{t})$, for some differential cocycle $\check{t}$ representing a torsion differential cohomology class. We then have 
\be
\check{b}_M = \pi_\ast\left( \frac{1}{2} \check{e}(T_V \tilde{M}) \cup \frac{1}{2}\check{e}(T_V \tilde{M})\right) + \check{t}
\ee
and we can pick $\check{t}$ so that $\check{b}_M = 0$. Equations \eqref{EqDefEffCn} and \eqref{EqDefEffC1} then simplify.

\subsection{Extension to manifolds of dimension 7 and 8}

\label{SecExt7-8-manifolds}

As reviewed in Section \ref{SecRemAnom}, the computation of the anomaly of a quantum field theory in dimension $d$ involves manifolds of dimension $d+1$ and $d+2$. Taking $X$ to be a 7- or 8-dimensional manifold, we endow it with a rank 5 bundle $\mathscr{N}_X$ satisfying \eqref{EqRelSWClassesTMN}. (From now on we will analogously write $\mathscr{N}_M$ for the normal bundle over $M$.) We then have a 4-sphere bundle $\tilde{X}$ over $X$ whose stable characteristic classes satisfy \eqref{EqRelCharClassMtMN}. As before, we write $\pi$ for the bundle map $\tilde{X} \rightarrow X$. 

$\check{C}_{\tilde{X}}$ is a differential cocycle on $\tilde{X}$ shifted by $\pi^\ast(w_4(TX \oplus \mathscr{N}_X))$. The constructions of Sections \ref{SecEffC-field} and \ref{SecStackM5} can be repeated on $X$, yielding differential cocycles $\check{a}_{\tilde{X}}$, $\check{A}_X$, $\check{b}_X$.

In the following, we will follow the notation in Section \ref{SecRemAnom} and write $U$ and $W$ for 7- and 8-dimensional manifolds, respectively. As we will argue below, the decoupling of the center of mass degrees of freedom on a stack of M5-branes requires a choice of global angular differential cocycle $\check{a}_{\tilde{M}}$, as introduced in \eqref{EqParCFieldTMStaM5}. It is therefore natural to consider the following. 6-dimensional closed smooth oriented Riemannian manifolds $M$, together with data $\mathfrak{d}_M = (\mathscr{N}_M, \check{C}_{\tilde{M}}, \check{a}_{\tilde{M}})$, can be seen as the objects of a cobordism category $\mathfrak{C}$, whose bordisms are oriented smooth Riemannian manifolds $U$ with boundary, together with data $\mathfrak{d}_U = (\mathscr{N}_U, \check{C}_{\tilde{U}}, \check{a}_{\tilde{U}})$. Of course, we require that if $(U,\mathfrak{d}_U)$ is a cobordism with boundary $(M,\mathfrak{d}_M)$, then $\mathfrak{d}_U|_{M} = \mathfrak{d}_M$. We also require that the Riemannian metric of $U$ is isomorphic to a direct product in a neighborhood of $\partial U$. Similarly, we will consider 8-dimensional cobordisms $(W,\mathfrak{d}_W)$ bounded by 7-dimensional closed manifolds $(U,\mathfrak{d}_U)$.

\subsection{Anomaly cancellation for non-intersecting M5-branes}

\label{SecAnCancDisM5}

M-theory on $Y$ contains a Chern-Simons term reading
\be
\label{EqM-thCSTerm}
{\rm CS}_{11} =  2\pi i \int_Y \left(\frac{1}{6} C \wedge G \wedge G - C \wedge I_8 \right) \;,
\ee
when the C-field is topologically trivial and can be represented by a 3-form $C$ with field strength $G$. The index density $I_8$ is defined in terms of the Pontryagin classes of $TY$ by
\be
I_8 = \frac{1}{48} \left(p_2(TY) + \left(\frac{p_1(TY)}{2}\right)^2\right) \;.
\ee
A more general formulation in terms of eta invariants can be found in \cite{Diaconescu:2003bm}. Alternatively, we can express it in differential cohomology. The integral Pontryagin cohomology class and the metric on $Y$ determine a differential cohomology class admitting $I_8$ as its field strength, which can be lifted to a differential cocycle $\check{I}_8$. In terms of the shifted differential cocycle $\check{C}$ describing the C-field, the Chern-Simons term \eqref{EqM-thCSTerm} can be written 
\be
\label{EqM-thCSTermDiffC}
{\rm CS}_{11} =  2\pi i \int_Y \left(\frac{1}{6} \check{C} \cup \check{C} \cup \check{C} - \check{C} \cup \check{I}_8 \right) \;,
\ee
where $\cup$ and $\int$ are the cup product and integral in differential cohomology \cite{hopkins-2005-70}. The integral of a differential cocycle of degree 12 on an 11-dimensional manifold gives an element of $\mathbb{R}/\mathbb{Z}$, reproducing the fact that the Chern-Simons term is defined only modulo $2\pi i$.

The Chern-Simons term \eqref{EqM-thCSTermDiffC} is a geometric invariant in the sense discussed in Section \ref{SecRemAnom}. In particular, it defines an anomaly line bundle over the base of families of 10-dimensional manifolds. When evaluated on a 11-dimensional manifold with boundary, it provides a section of this line bundle. As a result, when $Y$ has a boundary, \eqref{EqM-thCSTermDiffC} is not invariant under diffeomorphisms of and gauge transformations of the C-fields. There is both a gravitational and a gauge anomaly, which are canceled by the fields living on the boundaries of M-theory spacetimes \cite{Horava:1996ma}.

When the spacetime $Y$ has no boundaries but contains M5-branes wrapped on $M$, one is also naturally led to consider \eqref{EqM-thCSTermDiffC} on a manifold with boundary. As was already mentioned above, in this case the C-field, and therefore \eqref{EqM-thCSTermDiffC}, is defined only on $Y \backslash M$. Cutting out a small neighborhood $N$ of $M$, $CS_{11}$ needs to be evaluated on the manifold $Y \backslash N$, which has boundary $\tilde{M}$. This shows that the bulk action of M-theory has both gauge and gravitational anomalies in the presence of M5-branes. Those anomalies cancel against the anomalies present on the worldvolume of the M5-branes. This was discussed in \cite{Duff:1995wd,Witten:1996hc} and shown \cite{Freed:1998tg, Lechner:2001sj} for local anomalies. Global anomalies were shown to cancel in \cite{Monnierb}. 

For our purpose, this implies that in order to compute the anomaly associated to a system of (non-intersecting) M5-branes in some region of space, it is sufficient to evaluate the Chern-Simons term \eqref{EqM-thCSTermDiffC} on the boundary of a region containing them.

\section{Global anomalies of A-type (2,0) theories}

\label{SecGlobAnAn}

We compute in this section the global anomaly (2,0) theories in the A-series. In Section \ref{SecIdComp}, we introduce the scaling limit in which we obtain the (2,0) theory from a system of M5-branes. The computation of the anomaly of the stack of M5-branes is performed in Section \ref{SecEvCSTerm}. We then determine the anomaly of the center of mass tensor multiplet in Section \ref{SecGlobAnCM}, and deduce from it the global anomaly of the (2,0) theory in Section \ref{SecGlobAn20ThAn}. In Section \ref{SecConsCheck}, we check that the anomaly formula determines a well-defined geometric invariant of 7-manifolds. Section \ref{SecConfBlocks} presents the relation of the anomaly line bundle to the conformal blocks of the (2,0) theory and we discuss in Section \ref{SecOrigHWZTerms} a conceptual picture for the origin of the Hopf-Wess-Zumino terms present on the Coulomb branch of the (2,0) theory.

\subsection{Idea of the computation}

\label{SecIdComp}

We pick a compact smooth oriented 6-dimensional manifold $M$ and a rank 5 vector bundle $\mathscr{N}_M$ on $M$ whose Stiefel-Whitney classes satisfy \eqref{EqRelSWClassesTMN}. The total space of $\mathscr{N}_M$ is an oriented spin manifold, which we will see as an M-theory spacetime. We assume that $M$ carries a Riemannian metric and that $\mathscr{N}_M$ carries a connection. We endow $\mathscr{N}_M$ with a compatible metric. Inside $\mathscr{N}_M$, points at a fixed distance $R$ from the origin form a 4-sphere bundle $\tilde{M}$ over $M$.

We pick $n$ non-intersecting sections of $\mathscr{N}_M$ on which we wrap $n$ M5-branes. We assume that the largest distance between an M5-brane and the origin is $r$. 

As this system is formulated on a non-compact manifold, it displays a global anomaly under diffeomorphisms and gauge transformations that are not compactly supported. As explained in Section \ref{SecRemAnom}, the anomaly can be computed from a closed 7-manifold $U$. $U$ can be a mapping torus of $M$, if we are interested in computing the holonomy of the anomaly connection, or a twisted double, if we are interested in computing the anomalous phase of the partition function under a particular transformation. In any case, $U$ comes with the data $\mathfrak{d}_U = (\mathscr{N}_U, \check{C}_{\tilde{U}}, \check{a}_{\tilde{U}})$ extending the corresponding data on $M$ as described in Section \ref{SecExt7-8-manifolds}. We know from \cite{Monnierb} that the global anomaly vanishes in the bulk, so it can be computed by evaluating the M-theory Chern-Simons term on the asymptotic boundary of $\mathscr{N}_U$, which is diffeomorphic to $\tilde{U}$, the 4-sphere bundle over $U$ associated to $\mathscr{N}_U$.

We now take a decoupling limit in which we rescale both the Planck length $l_P$ and the fibers of $\mathscr{N}_M$, in a way such that $r/l_P^3$ stays constant \cite{Maldacena:1997re}. This limit is such that the M2-branes that might stretch between the M5-brane have constant energy. It ensures that the energy scale at which the gauge symmetry of the (2,0) theory is broken is constant. In the limit, we obtain effectively a free tensor supermultiplet describing the center of mass of the system, together with a (2,0) superconformal field theory of type $A_{n-1}$ at a generic point on its Coulomb branch. These theories are living on $M$, seen as the zero section of $\mathscr{N}_M$.

The global anomaly of the system does not change when we take the limit. As a consequence, we see that we can compute the global anomaly of the (2,0) superconformal field theory of type $A_{n-1}$ (together with the anomaly due to the center of mass) by evaluating the M-theory Chern-Simons term on $\tilde{U}$. Moreover, the anomaly has to be constant across the Coulomb branch. The computation to be performed below, a priori valid only at a generic point of the Coulomb branch, is therefore valid everywhere on the Coulomb branch.

\subsection{Evaluation of the Chern-Simons term}

\label{SecEvCSTerm}

After the limit described above has been taken, both the C-field and the metric on $\mathscr{N}_U$ are spherically symmetric. Moreover, the M-theory spacetime is empty away from the zero section. This implies that the Chern-Simons term can be evaluated on any round sphere bundle $\tilde{U} \subset \mathscr{N}_U$ centered around the origin. Taking $\tilde{U}$ to be a 4-sphere bundle with a finite radius avoids the slight complications coming from the fact that the metric blows up and the C-field field strength tends to zero as one approaches the asymptotic boundary of $\mathscr{N}_U$. Let us note that if $U$ is a mapping torus, adiabatic limits have to be taken in the formulas below. In the adiabatic limit, the metric along the base circle $c$ of $U$ blows up. To simplify the notation, we will suppress the adiabatic limits from the notation. No adiabatic limit is necessary in the case of most interest to us, when $U$ is a twisted double. 

We assume that $(U,\mathfrak{d}_U)$ is the boundary of $(W,\mathfrak{d}_W)$ (see Section \ref{SecExt7-8-manifolds}). The cobordism group computing the obstruction to the existence of $(W,\mathfrak{d}_W)$ has been described in Appendix C of \cite{Monnierb}. It is not known explicitly, but is at most torsion. To compute the anomaly of a stack of $n$ M5-branes, we need to evaluate 
\begin{align}
\label{EqEvCSTerm20}
{\rm An}_{nM5}(U) = & \, - \int_{\tilde{U}}  \left(\frac{1}{6} \check{C}_{\tilde{W}} \cup \check{C}_{\tilde{W}} \cup \check{C}_{\tilde{W}} - \check{C}_{\tilde{W}} \cup \check{I}_8\right) \\
= & \, - \int_{\tilde{W}}  \left(\frac{1}{6} G_{\tilde{W}} \wedge G_{\tilde{W}} \wedge G_{\tilde{W}} - G_{\tilde{W}} \wedge I_8\right) \notag \;,
\end{align}
where in the second line we expressed the Chern-Simons term on $\tilde{U}$ as the integral of the associated characteristic form on $\tilde{W}$. As explained in Section \ref{SecStackM5}, the C-field and its field strength on $\tilde{W}$ can be written
\be
\label{EqDecompGtW}
\check{C}_{\tilde{W}} = n \check{a}_{\tilde{W}} + \pi^\ast \check{A}_W \;, \quad G_{\tilde{W}} = n f_{\tilde{W}} + \pi^\ast F_W \;,
\ee
where $G_{\tilde{W}}$, $f_{\tilde{W}}$ and $F_W$ are the field strengths of $\check{C}_{\tilde{W}}$, $\check{a}_{\tilde{W}}$ and $\check{A}_W$, respectively.
$f_{\tilde{W}}$ integrates to 1 on the 4-sphere fibers of $\tilde{W}$. The term $n f_{\tilde{W}}$ in the field strength of the C-field comes from the fact that we have $n$ M5-branes at the origin sourcing the C-field. \eqref{EqDecompGtW} can be reparameterized as follows:
\be
\label{EqReparC-fieldtW}
\check{a}_{\tilde{W}} \rightarrow \check{a}_{\tilde{W}} + \pi^\ast(\check{B}_W) \;, \quad \check{A}_W \rightarrow \check{A}_W - n \check{B}_W \;,
\ee
for any degree 4 unshifted differential cocycle $\check{B}_W$. The minus sign in \eqref{EqEvCSTerm20} comes from the fact that the orientation of the boundary $\tilde{U}$ is reversed compared to \cite{Monnierb}. Equivalently \eqref{EqEvCSTerm20} yields directly the anomaly of the stack of M5-branes, as opposed to the anomaly inflow required to cancel it.

We now want to express \eqref{EqEvCSTerm20} as an integral on $W$. We can proceed as in Section 3.3 of \cite{Monnierb}. First, we see the integral on $\tilde{W}$ as the composition of a pushforward $\pi_\ast$ along the 4-sphere fibers with integration on $W$. The pushforward satisfies the relations 
\be
\label{EqRelPushForw}
\pi_\ast( \pi^\ast(x)) = 0 \;, \quad \pi_\ast(y \wedge \pi^\ast(x)) = \pi_\ast(y) \wedge x \;, \quad \pi_\ast(f_{\tilde{W}}) = 1 \;,
\ee
valid for differential forms $x \in \Omega^\bullet(W)$ and $y \in \Omega^\bullet(\tilde{W})$. The right-hand side of \eqref{EqEvCSTerm20} reads
\be
\label{EqIntFibSt1}
-\int_{W} \pi_\ast  \left(\frac{1}{6} (nf_{\tilde{W}} + \pi^\ast F_W)^3 - (nf_{\tilde{W}} + \pi^\ast F_W) \wedge I_8\right) \;.
\ee
Note that in this equation, the Pontryagin forms in $I_8$ are those of $T\tilde{W}$, and \eqref{EqRelCharClassMtMN} shows that they are the pull-back to $\tilde{W}$ of the Pontryagin forms of $TW \oplus \mathscr{N}_W$ on $W$. Using the latter fact and \eqref{EqRelPushForw}, we get
\be
\label{EqIntFibSt2}
-\int_{W}  \left(\frac{n^3}{6} \pi_\ast(f_{\tilde{W}}^3) + \frac{n^2}{2} \pi_\ast(f_{\tilde{W}}^2) \wedge F_W + \frac{n}{2}F_W^2 - nI_8 \right) \;,
\ee
where now $I_8$ is constructed out of the Pontryagin forms of $TW \oplus \mathscr{N}_W$. Next, we use the notation introduced in Section \ref{SecStackM5} to rewrite \eqref{EqIntFibSt2}:
\be
\label{EqIntFibSt3}
-\int_{W}  \left(n^3 \left(\frac{1}{6} \pi_\ast(f_{\tilde{W}}^3) - \frac{1}{8} \pi_\ast(f_{\tilde{W}}^2)^2\right) + \frac{n}{2}G_{W,n}^2 - nI_8 \right) \;.
\ee
The coefficient of $n^3$ is $\frac{1}{24}p_2(\mathscr{N}_W)$, as explained in Section 3.3 of \cite{Monnierb}. We further define the index density
\be
J_8 := I_8 - \frac{1}{24} p_2(\mathscr{N}_W)\;, 
\ee
computing the local anomaly of a free tensor multiplet, and we obtain 
\be
\label{EqAnnM53}
{\rm An}_{nM5}(U) = \int_W \left( nJ_8 - \frac{n^3-n}{24} p_2(\mathscr{N}_W) - \frac{n}{2}G_{W,n}^2 \right) \;.
\ee
Remark that $G_{W,n}$ is invariant under the reparameterization \eqref{EqReparC-fieldtW}, so \eqref{EqAnnM53} is manifestly invariant as well.

\subsection{The global anomaly of the center of mass}

\label{SecGlobAnCM}

\eqref{EqAnnM53} describes the global anomaly of the stack of $n$ M5-branes, corresponding to the (2,0) theory of type $A_n$ together with a free tensor supermultiplet of charge $n$, describing the center of mass of the system, as well as the degrees of freedom related to it by supersymmetry. In order to isolate the contribution from the (2,0) theory, we need to compute the global anomaly due to the free tensor multiplet.

To derive it, we temporarily ignore the fermions in the tensor multiplet, which do not have a gauge anomaly. The global anomaly of a self-dual field of charge 1 is given by \cite{hopkins-2005-70, Monniera}
\be
\label{EqAnnSD1}
{\rm An}_{SD,1}(U) = \int_W \left( \frac{1}{8}L(TW) - \frac{1}{2} G^2_W \right) \;,
\ee
where $L(TW)$ is the Hirzebruch genus of $TW$. $G_W$ is the field strength of a degree 4 differential cocycle $\check{C}_W$, modeling a 3-form gauge field coupling to the self-dual field. For the anomaly \eqref{EqAnnSD1} to be well-defined, in the sense discussed in Section \ref{SecIncons20TheoryAn}, it is crucial that $\check{C}_W$ is a differential cocycle shifted by the Wu class, as explained in Appendix \ref{SecProofInt}. Our aim is to separate the gravitational anomaly from the gauge anomaly in this expression. This is not a trivial problem, because although the first term in \eqref{EqAnnSD1} seems to capture the gravitational anomaly and the second one the gauge anomaly, they are not separately well-defined. For instance, the first term is obviously not an integer when evaluated on a closed manifold whose signature is not a multiple of 8. 

This problem can be cured by rewriting \eqref{EqAnnSD1} as 
\be
\label{EqAnnSD12}
{\rm An}_{SD,1}(U) = \frac{1}{8} \int_W (L(TM) - \sigma_W) -  \int_W \left(\frac{1}{2}G^2_W - \frac{1}{8}\sigma_W) \right) \;,
\ee
where $\sigma_W$ denotes the signature of the (non-degenerate) intersection form on the image of $H^4(W,\partial W; \mathbb{R})$ in $H^4(W; \mathbb{R})$. The point of this rewriting is that each of the two integrals yields an integer when evaluated on a closed manifold $W$, as explained in Appendix \ref{SecRelLiftWuClass}. Novikov's additivity theorem for the signature also ensures that the corresponding geometric invariants satisfy \eqref{EqFunctRelGeomInv}. Also, the dependence on the metric and on the C-field of the two terms remain unchanged compared to \eqref{EqAnnSD1}. We can therefore interpret the first term as the gravitational anomaly of the self-dual field, and the second one as the gauge anomaly, consistently with the detailed analysis of \cite{Monnier2011a, Monniera}. Both of these anomalies are well-defined in the sense of Section \ref{SecIncons20TheoryAn}.

The gravitational anomaly of a self-dual field of charge $n$ is the same as the one of a self-dual field of charge $1$, while its gauge anomaly is $n$ times larger. (More precisely, its gauge anomaly line bundle is the $n$th tensor product of the gauge anomaly line bundle of a self-dual field of charge 1. This implies that the holonomies and transition functions are taken to the $n$th power.) These facts determine the global anomaly of a self-dual field of charge $n$ to be 
\be
\label{EqAnnTMn1}
An_{SD,n}(U) =  \frac{1}{8} \int_W (L(TM) - \sigma_W) -  n \int_W \left(\frac{1}{2}G^2_W - \frac{1}{8}\sigma_W) \right) \;.
\ee
We deduce that the global anomaly of a tensor multiplet of charge $n$ is given by
\be
\label{EqAnnTMn2}
{\rm An}_{TM,n}(U) =  \int_W \left((J_8 + \frac{n-1}{8}\sigma_W - \frac{n}{2}G^2_W \right) \;.
\ee

\subsection{The global anomaly of the (2,0) theory}

\label{SecGlobAn20ThAn}

In \eqref{EqAnnTMn2}, $G_W$ is the field strength of a differential cocycle $\check{C}_W$ shifted by the Wu class. What is the differential cocycle that should be identified with $\check{C}_W$ when the tensor multiplet is the center of mass of a stack of M5-branes? It would be natural to set $\check{C}_W = \check{C}_{W,n}$, but this would be inconsistent, as $\check{C}_{W,n}$ is not shifted by the Wu class in general. The only $n$-independent cocycle with the correct shift in the problem is $\check{C}_W = \check{b}_W + \check{A}_W$. The fact that this cocycle is shifted by the Wu class was shown in Appendix B of \cite{Monnierb}, using crucial results of \cite{Witten:1999vg}.
 It was also argued in \cite{Monnierb} that $\check{C}_M  = \check{b}_M + \check{A}_M$ is the effective C-field coupling to the self-dual field on the worldvolume of a single M5-brane. It seems natural that the effective C-field coupling to the center-of-mass tensor multiplet should be given by the same expression.

Subtracting the contribution to the anomaly of the free center-of-mass tensor supermultiplet, we obtain a formula for the global anomaly of the (2,0) theory of type $A_{n-1}$:
\begin{align}
\label{EqAn20An}
& {\rm An}_{A_{n-1}}(U)  = {\rm An}_{nM5}(U) - {\rm An}_{TM,n}(U)  \\
& = \int_W \left( (n-1) J_8 - \frac{n^3-n}{24} p_2(\mathscr{N}_W) - \frac{n-1}{8}\sigma_W - \frac{n(n-1)}{2}h_W(2G_W + (n-1)h_W) \right) \notag \;.
\end{align}
where $G_W$ is the field strength of $\check{C}_W$.

As was discussed in Section \ref{SecStackM5}, if $w_4(\mathscr{N}_M) = 0$, there is a preferred choice for the global angular cocycle $\check{a}_M$, which results in $\check{b}_M = 0$. If the extensions of the normal bundle are such that $w_4(\mathscr{N}_U) = w_4(\mathscr{N}_W) = 0$, then we can extend the global angular cocycle on $\tilde{U}$ and $\tilde{W}$ in such a way that $\check{b}_U = \check{b}_W = 0$. In particular, $h_W = 0$ and the last term vanishes. This is for instance the case when $\mathscr{N}_M$ is trivial. However the cases where $\mathscr{N}_M$ is non-trivial are very important, as they correspond to twistings of the (2,0) theory. Then, even if $w_4(\mathscr{N}_M) = 0$, there is in general no reason that would force $w_4(\mathscr{N}_U) = w_4(\mathscr{N}_W) = 0$ for all the twisted doubles $U$. In fact, we will see that the last term is crucial for the consistency of \eqref{EqAn20An}.

It is interesting to note that there remains a dependence on the background C-field, through the extension $G_W$ of its field strength to $W$. There is as well a dependence on $h_W$, the field strength of \eqref{EqDefb}, and therefore a dependence on the choice of parameterization \eqref{EqDecompGtW}. These somewhat puzzling features can all be traced back to the decoupling of the center of mass degrees of freedom. This operation requires picking a differential cocycle of degree 3 shifted by the Wu class, which is the effective C-field on the worldvolume coupling to the center-of-mass tensor multiplet. There is no way to do this canonically and the choice we made, $\check{C}_M$, extended to $W$ as $\check{C}_W$, depends on \eqref{EqDecompGtW}. In contrast, the anomaly formula \eqref{EqAnnM53} for a stack of M5-branes, including the center of mass, is independent of \eqref{EqDecompGtW}.

A consequence of this fact is that the decomposition \eqref{EqDecompGtW} cannot be chosen freely on $W$. The definition of the (2,0) theory on $M$ should include a choice global angular differential cocycle $\check{a}_{\tilde{M}}$ on $\tilde{M}$, which should then be extended to $\tilde{U}$ and $\tilde{W}$, as was already suggested in our discussion of Section \ref{SecExt7-8-manifolds}. A choice of $\check{a}_{\tilde{M}}$ is effectively a choice of a vertical cotangent bundle on $\tilde{M}$. It is therefore not so surprising that when the normal bundle $\mathscr{N}_M$ is topologically non-trivial, such a choice has to be made in order to decouple the center of mass, and that this choice cannot be made canonically.

As we are only interested in the (2,0) theory, we should set the C-field on $M$ to a preferred value, for instance zero. Because of an analog of the Freed-Witten anomaly for self-dual fields, first described in \cite{Witten:1999vg}, this might not be consistent. We should rather set $\check{C}_M = \check{S}_M$, where $\check{S}_M$ is a certain 2-torsion differential cocycle determined by the anomaly cancellation condition (see Section 3.6 of \cite{Monniera}). Together with a choice $\check{a}_{\tilde{M}}$ of a global angular cocycle on $\tilde{M}$, this fixes the value of the M-theory C-field on $\tilde{M}$.

We can recover the local anomaly from \eqref{EqAn20An} by taking $U$ to be a mapping torus over a small homotopically trivial loop $c$ in the space of background fields. The holonomy of the anomaly connection along $c$ is then proportional to the value of its curvature inside the loop. In this case, we can take $W = M \times D^2$, $\tilde{W} = \tilde{M} \times D^2$, where $D^2$ is a 2-dimensional disk. As the metric alone is changing along $D^2$, only the metric-dependent terms can have a non-zero integral. But the only metric-dependent terms are the first two in \eqref{EqAn20An}. A comparison with \cite{Harvey:1998bx} (see also \eqref{EqLocAn20Theory}) shows that these two terms reproduce the index density governing the local anomaly derived in that paper. Let us also remark that in \cite{Harvey:1998bx}, it was assumed that the local gravitational anomaly cancellation, proven for a single M5-brane, holds as well for a stack of M5-branes. Our derivation requires no such assumption. We rather relied on the cancellation of global anomalies for non-intersecting M5-branes, proven in \cite{Monnierb} to deduce the anomaly at a generic point on the Coulomb branch. 

We will also see in the next section that the last two terms in \eqref{EqAn20An}, while having no effect on the local anomaly, are crucial for the anomaly to be consistent globally.

%We should emphasize that \cite{Monnierb} and the present paper assume that the closed 7-manifold $(U,\mathfrak{d}_U)$ is the boundary of a manifold $(W,\mathfrak{d}_W)$. This means that \eqref{EqAn20An} determines ${\rm An}_{A_{n-1}}(U)$ only on those $(U,\mathfrak{d}_U)$ that are boundaries. It is not known whether there are elements $(U,\mathfrak{d}_U)$ that are not boundaries, but their cobordism classes are at most torsion, as shown in Appendix C of \cite{Monnierb}. This problem could be avoided by performing the computation above directly on $U$, using the first line of \eqref{EqEvCSTerm20}. This should be possible, but we will not attempt it here. 

\subsection{A consistency check}

\label{SecConsCheck}

In this section we check that when \eqref{EqAn20An} is evaluated on a closed manifold $W$, it yields an integer. This ensures that the anomaly is well-defined, in the sense discussed in Section \ref{SecIncons20TheoryAn}. Strictly speaking, this check is not necessary. We obtained \eqref{EqAn20An} as the difference of two terms describing well-defined anomalies. One is the reduction of the characteristic form associated to the M-theory Chern-Simons term, which takes integral values on closed manifolds as shown in \cite{Witten:1996md}. The other is the global anomaly of the center of mass, which is shown in Appendix \ref{SecRelLiftWuClass} to take integral values on closed manifolds as well. Nevertheless, this is a good check on our computations and it involves some interesting algebraic topology. 

In the rest of this section, $W$ is a closed oriented 8-manifold. Let us first remark that the analysis of the cancellation of local anomalies for five-branes \cite{Witten:1996hc, Freed:1998tg} shows that $J_8 = \frac{1}{8}L(TW) - \frac{1}{2}I_f$, where $I_f$ is the index density of the chiral fermions on the worldvolume of a single M5-brane. As the Dirac operator associated to $I_f$ is quaternionic on an 8-dimensional manifold (see Section 3.1 of \cite{Monnierb}), its index is even and $\int_W \frac{1}{2}I_f$ is an integer. The term involving the Hirzebruch genus integrates to the signature of $W$ and cancels with the third term in \eqref{EqAn20An}. Therefore, all that remains to be shown is that the second and fourth terms in \eqref{EqAn20An} add up to an integer. 

To this end, it is useful to distinguish two cases, depending on whether $n$ is even or odd. For odd $n$, $n^3-n$ is a multiple of 24 (it is sufficient to check this explicitly for $n = 1$ to $23$), so the second term is an integer. To see that the last term is an integer as well in this case, we let $n = 2k+1$ and write it
\be
\label{Eq4thTermnOdd}
\int_W \frac{2k+1}{2} 2kh_W(2G_W + 2kh_W) \;.
\ee
But $2kh_W$ is a closed form with integral periods. $2G_W$ is a closed form with integral periods as well, but in addition it is a form lift of the Wu class (see Appendix \ref{SecRelLiftWuClass}). This implies that $2G_W$ is a characteristic element for the wedge product pairing on the space of closed forms on $W$ with integral periods, which implies that \eqref{Eq4thTermnOdd} is an integer.

In case $n$ is even, we need more sophisticated tools. Again, a straightforward inspection shows that for $n = 2k$ even, 
\be
4 \frac{n^3-n}{24} =  k \quad {\rm mod} \; 4 \;.
\ee
On the other hand, the fourth term in \eqref{EqAn20An} reads
\begin{align}
- k(2k-1) & \, \int_W h_W(2G_W + (2k-1)h_W) \\
 & \, = - \frac{(2k-1)}{2} \int_W 2kh_W (2G_W + 2kh_W) + k(2k-1) \int_W h_W^2 \notag\;.
\end{align}
For the same reason as above, the first term on the right-hand side is an integer, and as $h_W$ has half-integral periods, the second term belongs to $\mathbb{Z}/4$. As $k(2k-1) = k$ mod $4$, all we need to show is that $\int_W 4h_W^2 = \int_W p_2(\mathscr{N}_W)$ mod $4$.

For this, we need to introduce a cohomological operation, the Pontryagin square $\mathfrak{P}$. $\mathfrak{P}$ maps $H^\bullet(W;\mathbb{Z}_2)$ into $H^\bullet(W;\mathbb{Z}_4)$. Denoting by $\rho_k$ the reduction modulo $k$, the Pontryagin square has the property that $\mathfrak{P} \rho_2(u) = \rho_4(u^2)$ for any $u \in H^\bullet(W;\mathbb{Z})$. The action of the Pontryagin square on Stiefel-Whitney classes has been computed by Wu \cite{Wu1959} and can be found for instance in \cite{Thomas1960}:
\be
\mathfrak{P}(w_{2i}) = \rho_4(p_i) + \theta_2 \left( w_1 Sq^{2i-1} w_{2i} + \sum_{j = 0}^{i-1} w_{2j} w_{4i-2j} \right) \;.
\ee
In this formula, $Sq^i$ are the Steenrod squares and $\theta_2$ is the embedding of $H^\bullet(W;\mathbb{Z}_2)$ into $H^\bullet(W;\mathbb{Z}_4)$ induced by the corresponding embedding of $\mathbb{Z}_2$ into $\mathbb{Z}_4$. Applying this formula to the bundle $\mathscr{N}_W$, we see that $\mathfrak{P}(w_4(\mathscr{N}_W)) = p_2(\mathscr{N}_W)$ mod $4$, as $w_i(\mathscr{N}_W) = 0$ for $i > 5$. But now we can use the fact that $2h_W$ is a form lift of $w_4(\mathscr{N}_W)$, i.e. the periods of $2h_W$ on 4-cycles on $W$ are even or odd depending on whether $w_4(\mathscr{N}_W)$ has period 0 or 1. Together with the property of $\mathfrak{P}$ mentioned above, this implies that 
\be
\int_W 4h_W^2 = \int_W p_2(\mathscr{N}_W) \quad {\rm mod} \; 4 \;.
\ee
We have therefore shown that \eqref{EqAn20An} always takes integer values on closed manifolds $W$. The somewhat strange-looking fourth term is essential in order to cure the ambiguities of the second term.

\subsection{The conformal blocks}

\label{SecConfBlocks}

A potentially confusing point is the following. The geometric invariant ${\rm An}_{A_{n-1}}$ defines a line bundle $\mathscr{L}_{A_{n-1}}$ over the space of objects of the cobordism category $\mathfrak{C}$, that is over the space of 6-manifolds $M$ endowed with the data $\mathfrak{d}_M$. We expect the partition function of the (2,0) theory to be a section of this line bundle. 

But it is known that the (2,0) theory does not admit a single partition function. Rather, it has a space of ``conformal blocks'' whose dimension is given by the order of Lagrangian subgroups of $H^3(M;\mathbb{Z}_n)$ with respect to the cup product pairing on $H^3(M;\mathbb{Z}_n)$ \cite{Witten:1998wy, Witten:2009at}.

These two statements can be reconciled as follows. The partition function $Z_{n{\rm M5}}$ of a stack of M5-branes is well-defined and unique. The conformal blocks arise after the decoupling of the center-of-mass tensor multiplet, because the self-dual field of charge $n$ that it contains does not have a single partition function, but rather a set of conformal blocks $Z_{{\rm CM},x}$ \cite{Witten:1998wy}. They form a representation of a central extension $G_H$ of $H^3(M;\mathbb{Z}_n)$ and can be parameterized by an index $x$ running over a Lagrangian subgroup of $H^3(M;\mathbb{Z}_n)$. As $Z_{n{\rm M5}}$ is invariant under $G_H$ and $Z_{{\rm CM},x}$ transforms in the irreducible unitary representation of $G_H$, it is natural to expect that the (2,0) theory has conformal blocks $Z_{A_{n-1},x}$ valued in the dual representation, and that one can write $Z_{n{\rm M5}} = \sum_x Z_{{\rm CM},x}Z_{A_{n-1},x}$. Similar statements in the case of $N=4$ super Yang-Mills were put forward in \cite{Belov:2004ht}. Now $Z_{{\rm CM},x}$ are all sections of the same line bundle. In order for the sum to make sense, the conformal blocks $Z_{A_{n-1},x}$ should all be sections of a unique line bundle; this is the line bundle $\mathscr{L}_{A_{n-1}}$.

The fact that $Z_{{\rm CM},x}$ are sections of the same line bundle for all $x$ also justifies our computation of the anomaly of the (2,0) theory in Section \ref{SecGlobAn20ThAn} by subtracting the anomaly of the center-of-mass tensor multiplet from the anomaly of the stack of M5-branes.  

In more detail, recall that we can parameterize the M-theory C-field on $\tilde{M}$ as follows
\be
\label{EqParC-fieldOnM}
\check{C}_{\tilde{M}} = n\check{a}_{\tilde{M}} + \pi^\ast(\check{A}_M) \;.
\ee
Clearly, the differential cohomology class of $\check{C}_{\tilde{M}}$ is left invariant under shifts 
\be
\label{EqReparC-fieldtM}
\check{a}_{\tilde{M}} \rightarrow \check{a}_{\tilde{M}} + \pi^\ast(\check{B}_M) \;,
\ee
where $\check{B}_M$ is a differential cocycle on $M$ representing an order $n$ differential cohomology class. (From now on, we will make a slight abuse of language and refer to $\check{B}_M$  as an ``order $n$ torsion differential cocycle'', even if $n\check{B}_M$ is zero only in cohomology.) The effective C-field to which the center-of-mass tensor multiplet couples is
\be
\check{C}_M = \frac{1}{2}\pi_\ast(\check{a}_{\tilde{M}} \cup \check{a}_{\tilde{M}}) + \check{A}_M \;,
\ee
transforming as: 
\be
\label{EqTransCM}
\check{C}_M \rightarrow \check{C}_M + \check{B}_M \;.
\ee
So the differential cohomology class of $\check{C}_M$ is not invariant under such changes of parameterization. The transformation \eqref{EqReparC-fieldtM} acts on the conformal blocks of the center of mass, which are functions of $\check{C}_M$, but leaves $Z_{n{\rm M5}}$ invariant. 

At least if there is no torsion in $H^3(M;\mathbb{Z})$, we can be more precise. In this case, a (linearly dependent) set of generators of the conformal blocks of the center of mass is provided by level $n$ Siegel theta functions over the torus $\mathcal{J}_n$ of flat (gauge equivalence classes of) C-fields \cite{Henningson:2010rc}. The latter is defined by $\mathcal{J}_n = H^3(M;\mathbb{R})/nH^3_{\mathbb{Z}}(M;\mathbb{Z})$, where $H^3_{\mathbb{Z}}(M;\mathbb{R})$ denotes the de Rham cohomology classes having integral periods on $M$. \eqref{EqTransCM} is then simply an order $n$ rotation of $\mathcal{J}_n$. It is well-known that the theta functions of level $n$ are in bijection with order $n$ points of $\mathcal{J}_n$, and therefore \eqref{EqTransCM} simply permutes the elements in our set of conformal blocks. If torsion is present, the space of flat C-fields $\check{H}^4_{\rm flat}(M)$ fit in a short exact sequence
\be
0 \rightarrow \mathcal{J}_n \rightarrow \check{H}^4_{\rm flat}(M) \rightarrow H^4_{(n)}(M;\mathbb{Z}) \;,
\ee
where $H^4_{(n)}(M;\mathbb{Z})$ is the subgroup generated by the elements of order $n$ in $H^4(M;\mathbb{Z})$. The order $n$ differential cocycle $\check{B}_M$ then acts on $\check{H}^4_{\rm flat}(M)$ by order $n$ rotations of the components $\mathcal{J}_n$ together with permutations of these.

In summary, the data $\mathfrak{d}$ defined in Section \ref{SecExt7-8-manifolds} is the data required to define the (2,0) theory \emph{and select a particular conformal block}. All the conformal blocks of the (2,0) theory are sections of the same line bundle over the moduli space of manifolds endowed with the data $\mathfrak{d}$. This line bundle is determined by ${\rm An}_{A_{n-1}}$ as explained in Section \ref{SecAnomCob}. The conformal blocks share the same anomaly and are permuted by the shifts \eqref{EqReparC-fieldtM} of $\check{a}_{\tilde{M}}$.

In contrast, the data required to define the $A_{n-1}$ (2,0) theory without a choice of conformal block is (keeping the notation of Section \ref{SecExt7-8-manifolds}) $\mathfrak{d}_M' = (\mathscr{N}_M, \check{C}_{\tilde{M}}, n\check{a}_{\tilde{M}})$, where now $\check{a}_{\tilde{M}}$ is determined up to a torsion element of order $n$. Over the moduli space of manifolds with data $\mathfrak{d}'$, the conformal blocks should rather be seen as sections of a vector bundle, whose rank is given by the order of Lagrangian subspaces of $H^3(M;\mathbb{Z}_n)$. To describe the anomaly precisely in this context requires to promote the geometric invariant ${\rm An}_{A_{n-1}}$ to an anomaly field theory \cite{Freed:2014iua}. The relevant anomaly field theory is a type of quantum Dijkgraaf-Witten theory, whose classical version is given by ${\rm An}_{A_{n-1}}$ and whose quantization is performed by summing over the torsion component of $\check{a}_{\tilde{M}}$. The details of this construction will appear in a future paper \cite{Monnierc}. 

This generalization is important because there exist diffeomorphisms that fail to preserve the torsion component of $\check{a}_{\tilde{M}}$. Such diffeomorphisms permute the conformal blocks of the (2,0) theory and their action cannot be accounted for naturally using the formalism developed in the present paper. \footnote{We thank the referee for making this point.} Indeed, they were implicitly ruled out by the choice of the data $\mathfrak{d}$, which they fail to preserve.

Let us also remark also that the picture developed in this section shows that all the subtleties of the (2,0) theory at a non-generic point on its Coulomb branch are captured by the partition function $Z_{n{\rm M5}}$ of the stack of M5-branes and are independent of the choice of conformal block. 

\subsection{The origin of the Hopf-Wess-Zumino terms}

\label{SecOrigHWZTerms}

A naive computation of the local gravitational anomaly of the (2,0) $A_{n-1}$ theory by summing the anomalies of the $n$ tensor multiplets present at a generic point on the Coulomb branch fails to capture the whole anomaly of the theory. It was proposed in \cite{Intriligator:2000eq} that the effective theory on the Coulomb branch contains certain Wess-Zumino terms, dubbed ``Hopf-Wess-Zumino terms'', compensating for the difference between the naive computation and the correct anomaly found in \cite{Harvey:1998bx}. In our framework, those terms are responsible for the second and fourth terms of the anomaly \eqref{EqAn20An}, although only the second term was accounted for in \cite{Intriligator:2000eq}. We show here that these Wess-Zumino terms can be pictured very concretely as the topological modes of the M-theory C-field that get trapped between the M5-branes when the decoupling limit of Section \ref{SecIdComp} is taken. A somewhat similar idea was mentioned in \cite{Kalkkinen:2002tk}.

Recall our method to compute the anomaly inflow in Section \ref{SecIdComp}. We considered a set of $n$ non-intersecting M5-branes separated by a typical distance $r$. We picked a tubular neighborhood $N_0$ of $M$ including all the M5-branes, say of radius $R_0$. We then rescaled $r$ to zero while keeping $R_0$ fixed. Equivalently, we could have kept $r$ fixed and taken $R_0$ to infinity. 

% ---- DO NOT ERASE, FIGURE TO BE INCLUDED IN THE PDF VERSION ----
\begin{figure}[htb]

  \centering
  \includegraphics[width=\textwidth]{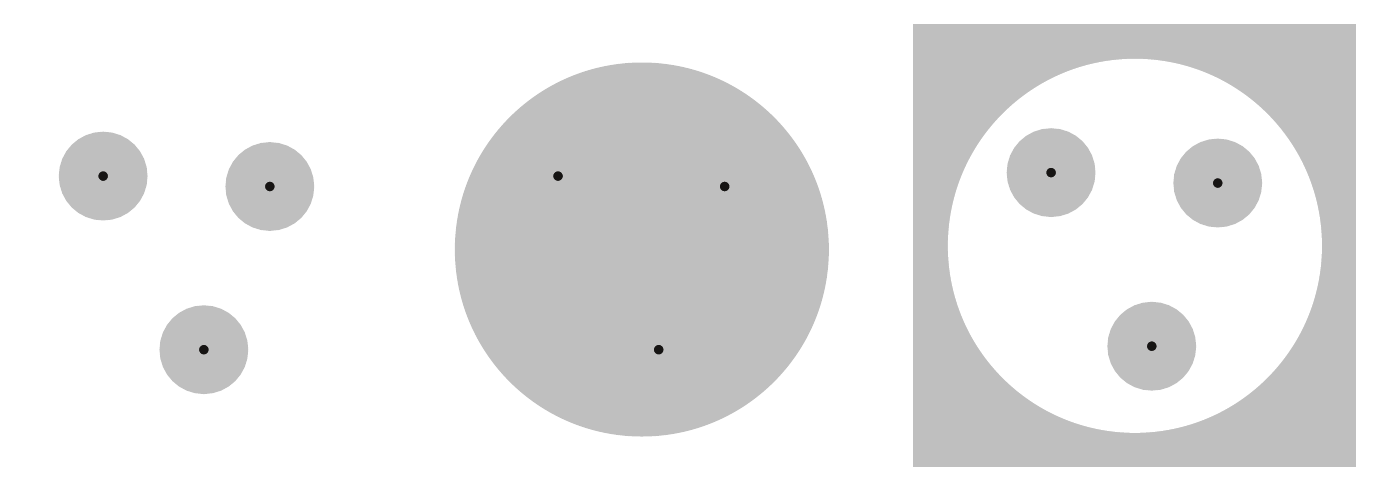}
  \caption{\emph{A pictorial representation of the arguments in this section. The three pictures represent a fiber over a point of $M$. On the left, the setup used to compute the anomaly due to a set of non-intersecting M5-branes (black dots). Tubular neighborhoods (grayed out) are cut out and there is an anomaly inflow from the M-theory Chern-Simons term in the bulk (in white). This inflow cancels exactly the sum of the anomalies of the isolated M5-branes. 
  In the middle, the setup presented in Section \ref{SecIdComp} in order to compute the anomaly of a stack of M5-branes on its Coulomb branch. A single tubular neighborhood of $M$ is cut out and includes all the M5-branes. Again, there is an anomaly inflow due to the M-theory Chern-Simons term in the bulk.
    On the right, the difference between the anomaly inflow contributions can be attributed to the M-theory Chern-Simons term integrated over the region $N$, represented in white.}}
   \label{Fig1}
\end{figure}
% ---- DO NOT ERASE, FIGURE TO BE INCLUDED IN THE PDF VERSION ----

An alternative way of computing the anomaly is the following. We take $n$ non-intersecting tubular neighborhoods $N_i$ of the worldvolumes $M_i$ of each M5-brane, of radius $R_i << r$. Let us write $\tilde{M}_i = \partial N_i$, a 4-sphere bundle over $M_i$. If this setup is extended to a 7-manifold $U$, we can compute the inflow due to the bulk on this system by evaluating the M-theory Chern-Simons term on $\bigcup_i \tilde{U}_i$ and taking a limit in which $R_i$ scale down to zero. It is clear that the anomaly obtained in this way is the sum of the anomalies due to each M5-brane. In other words, via this procedure, we obtain the naive anomaly mentioned at the beginning of this section.

But now the reason why the two procedures do not give the same answer is clear. In the first procedure, in addition to the M5-branes themselves, we also included a part of the bulk of M-theory, namely 
\be
N := N_0 \backslash \bigcup_{i=1}^n N_i \;.
\ee
The M-theory Chern-Simons term on $N$ is anomalous, because $N$ has boundaries. In fact, when $M$ is promoted to a 7-manifold $U$, the anomaly due to the Chern-Simons term can be obtained by evaluating it on $\tilde{U} \cup \bigcup_i (\overline{\tilde{U}}_i)$. We see that the anomaly difference between a stack of M5-branes on its Coulomb branch and a set of non-intersecting M5-branes is entirely due to the M-theory Chern-Simons term on $N$. See Figure \ref{Fig1}.

$N$ is a fiber bundle over $M$. The fiber is a 5-ball of radius $R$ out of which $n$ 5-balls of radii $R_i$ have been carved out. Writing $\pi$ for the bundle map and $\check{cs}_{11}$ for the integrand of \eqref{EqM-thCSTermDiffC}, the Hopf-Wess-Zumino term is 
\be
\
\check{\rm wz} = \pi_\ast({\rm \check{cs}}_{11}) \;,
\ee
i.e. the integral of the Chern-Simons integrand over the fibers of $N$, yielding a top differential cocycle on $M$. By definition, we have
\be
\int_M \check{\rm wz} = \int_N {\rm \check{cs}}_{11} \;,
\ee
and $\check{\rm wz}$ is a local term on $M$ accounting for the anomaly difference.
Finally, we have to take the limit $R_0 \rightarrow \infty$, $R_i \rightarrow 0$. The advantage of this formulation is that it is completely general: no assumption is made on the topology of the system of M5-branes, except that they are not intersecting. Of course, in order to get an explicit expression for the Hopf-Wess-Zumino terms, the setup should be simple enough so that the integration over the fibers of $N$ can be performed explicitly.

\section{The global anomaly of a generic A-D-E (2,0) theory}

\label{SecGenGlobAnForm}

In the present section, we show that the anomaly we found for A-type theories can be naturally rewritten in terms of basic Lie algebra data. This result yields a conjectural formula for the global anomaly of a generic (2,0) theory, which is automatically compatible with the exceptional isomorphisms among the A-D-E Lie algebras. We also provide a consistency check by showing that the corresponding anomaly is well-defined in the sense of Section \ref{SecIncons20TheoryAn}. 

\subsection{The anomaly formula}

\label{SecAnFormADE}

For a simply laced simple Lie algebra $\mathfrak{g}$, the general global anomaly formula reads
\be
\label{EqGenAnFormADE}
{\rm An}_{\mathfrak{g}}(U) = \int_W \left( r(\mathfrak{g}) J_8 - \frac{|\mathfrak{g}|{\rm h}_\mathfrak{g}}{24} p_2(\mathscr{N}_W) - \frac{r(\mathfrak{g})}{8}\sigma_W - r(\mathfrak{g}) {\rm h}_{\mathfrak{g}} h_W \left(G_W - h_W\right) - \frac{|\mathfrak{g}|{\rm h}_\mathfrak{g}}{2} h_W^2 \right) \;,
\ee
where $|\mathfrak{g}|$ is the dimension of $\mathfrak{g}$, $r(\mathfrak{g})$ its rank and ${\rm h}_\mathfrak{g}$ its dual Coxeter number. \eqref{EqGenAnFormADE} coincides with \eqref{EqAn20An} for the $A$-type theory. Note that as \eqref{EqGenAnFormADE} is expressed in terms of data intrinsic to $\mathfrak{g}$, this formula is automatically compatible with the exceptional isomorphisms among elements of the $A$, $D$ and $E$ series. For the reader's convenience, we recall the values of the dimension and of the dual Coxeter numbers:
\be
\begin{array}{l|ll}
& |\mathfrak{g}| & {\rm h}_\mathfrak{g} \\ \hline
A_n & n^2 + 2n & n + 1 \\
D_n & 2n^2 - n & 2n - 2 \\
E_6 & 78 & 12 \\
E_7 & 133 & 18 \\
E_8 & 248 & 30 \\
\end{array} \;.
\ee
Of course, the rank of $X_n$ is $n$. The first two terms of \eqref{EqGenAnFormADE}, which are the only ones relevant for the local anomaly, were obtained in \cite{Intriligator:2000eq}.

\subsection{Data required to specify a (2,0) theory}

As we already discussed, in the A-type theories, $h_W$ and $G_W$ have a clear interpretation in terms of M-theory data. For the other (2,0) theories, it is not obvious how these objects should be interpreted, especially for the E-type theories, where there is no M-theory realization. We define here data on $M$ that naturally give rise to $h_W$ and $G_W$. Presumably, this data is required in order to define the (2,0) theory on a 6-manifold $M$, independently of any M-theory realization.

We already know that in order to define the (Euclidean) (2,0) theory, we need a simply laced Lie algebra $\mathfrak{g}$, a smooth oriented Riemannian manifold $M$, an R-symmetry bundle $\mathscr{N}_M$ satisfying \eqref{EqRelSWClassesTMN} and a spin structure on $TM \oplus \mathscr{N}$. We claim that in addition to this we need a choice of global angular differential cocycle $\check{a}_{\tilde{M}}$ on the 4-sphere bundle $\tilde{M}$ associated to $\mathscr{N}_M$. 

We saw that in the A-type theories, such a choice was necessary in order to perform the decoupling of the center-of-mass degrees of freedom. $\check{a}_{\tilde{M}}$, together with the requirement $\check{C}_M = \check{S}_M$, fully determines the M-theory C-field on $\tilde{M}$. Similarly, in any (2,0) theory, a choice of $\check{a}_{\tilde{M}}$ allows one to define $\check{b}_M := \frac{1}{2} \pi_\ast(\check{a}_{\tilde{M}} \cup \check{a}_{\tilde{M}})$, $\check{C}_M = \check{S}_M$ and $\check{A}_{M} = \check{C}_M - \check{b}_M$. In anomaly computations, this data is extended to 7- and 8-dimensional manifolds $U$ and $W$. $h_W$ and $G_W$ in \eqref{EqGenAnFormADE} are then respectively the field strengths of $\check{b}_W$ and $\check{C}_W$.

\subsection{Consistency}

Using our analysis of the $A_n$ case, it is easy to see that \eqref{EqGenAnFormADE} yields an integer on closed manifolds for any $\mathfrak{g}$, and therefore that it describes a well-defined anomaly. Indeed, the following terms take independently integer values on closed manifolds:
\be
\label{EqIntTerGenForm1}
\int_W \left(r(\mathfrak{g}) J_8 - \frac{r(\mathfrak{g})}{8}\sigma_W \right) \;, 
\ee
\be
\label{EqIntTerGenForm2}
\int_W \left( \frac{|\mathfrak{g}|{\rm h}_\mathfrak{g}}{24} p_2(\mathscr{N}_W) + \frac{|\mathfrak{g}|{\rm h}_\mathfrak{g}}{2} h_W^2   \right) \;, 
\ee
\be
\label{EqIntTerGenForm3}
\int_W r(\mathfrak{g}) {\rm h}_{\mathfrak{g}} h_W (G_W - h_W) \;.
\ee
The fact that \eqref{EqIntTerGenForm1} is an integer was explained in Section \ref{SecGlobAnCM}. To show that \eqref{EqIntTerGenForm2} is an integer, recall that we proved in Section \ref{SecConsCheck} that $\int_W p_2(\mathscr{N}_W) = 4\int_W h_W^2$ mod 4. Integrality will follow provided that $|\mathfrak{g}|{\rm h}_\mathfrak{g}/6 = -|\mathfrak{g}|{\rm h}_\mathfrak{g}/2$ mod 4, which requires $|\mathfrak{g}|{\rm h}_\mathfrak{g}$ to be a multiple of 6. This can be readily checked in each case. Finally, the last term takes integer value because $r(\mathfrak{g}) {\rm h}_{\mathfrak{g}}$ is even, $2G_W$ and 
$2h_W$ have integral periods, and $2G_W$ is a lift of the Wu class, hence is a characteristic element of the wedge product pairing on forms with integral periods (see Appendix \ref{SecRelLiftWuClass}). 

\subsection{Further comments}

We do not have a compelling picture explaining how the conformal blocks arise in D- and E-type theories.  

It would be interesting to derive the anomaly formula \eqref{EqGenAnFormADE} from the type IIB realization of the (2,0) theories, but we leave this for future work. 

We attempted to derive \eqref{EqGenAnFormADE} for the $D_n$ series using M-theory on an $\mathbb{R}^5/\mathbb{Z}_2$ orbifold. However we cannot perform a rigorous derivation, because of a puzzling feature of the orbifold background: the anomaly of the orbifold is not well-defined globally. This can be understood from the fact that the $\mathbb{R}^5/\mathbb{Z}_2$ sources a half-quantum of flux of the M-theory C-field. The orbifold singularity has an anomaly ``${\rm An}_{O}(U) = -\int_W \frac{1}{2}J_8$'' canceled by anomaly inflow. But as $\frac{1}{2}J_8$ does not integrate to an integer on a closed manifold $W$, the expression above does not define a geometric invariant of $U$. We therefore encounter the same problem that was plaguing the naive anomaly formula \eqref{EqGuessAn20Theory} for the (2,0) theory, and unlike in the latter case, there seems to be no extra term appearing to cure the inconsistency. Closing our eyes to this problem, a calculation very similar to that for $A_n$ theory yields all the terms in \eqref{EqGenAnFormADE} with the right prefactors, except for the fourth one. Because of this, the anomaly derived in this way is inconsistent. We expect that a proper understanding of the orbifold's anomaly should cure this problem.

\subsection*{Acknowledgments}

This research was supported in part by SNF Grant No.200020-149150/1.

\appendix

\section{Properties of lifts of the Wu class}

\label{SecRelLiftWuClass}

We review here some basic properties of the Wu class and its lifts, which play an important role in the proofs of the paper.

\subsection{The Wu class and its lifts}

Recall that the Wu class on a closed manifold $X$ of dimension $d$ is an element $\nu = \sum_k \nu_k$ of $H^\bullet(X;\mathbbm{Z}_2)$ satisfying 
\be
\langle Sq^k(x), [X] \rangle = \langle x \cup \nu_k, [X] \rangle
\ee
for $x \in H^{d-k}(X;\mathbbm{Z}_2)$. $Sq$ denotes here the Steenrod operations and $[X]$ is the fundamental homology class of $X$. In case the dimension of $X$ is even and $x$ is of degree $d/2$, $Sq^{d/2}(x) = x \cup x$ and the above reduces to $x \cup x = x \cup \nu_{d/2}$. $\nu$ can be expressed in terms of the Stiefel-Whitney classes. For instance, on an oriented manifold, $\nu_4 = w_4 + w_2^2$.

We call a closed differential form $\lambda \in \Omega^k(X)$ a \textit{form lift of the Wu class} if the periods of $\lambda$ are integers and equal to the periods of $\nu_k$ modulo 2. Let $\check{C}$ be a differential cocycle shifted by the Wu class on $X$ (see Section 2.1 of \cite{Monnierb}) and let $G$ be its field strength. Then $2G$ is a form lift of the Wu class.

\subsection{Proof of integrality}

\label{SecProofInt}

Let $X$ be of even dimension $d$ and let $\lambda$ be a form lift of the Wu class of degree $d/2$. Then $\lambda$ is a characteristic element for the wedge product pairing on the space $\Omega^{d/2}_\mathbb{Z}(X)$ of closed forms with integral periods, namely
\be
\label{EqPropLiftWu}
\int_X F \wedge F = \int_X F \wedge \lambda \quad {\rm mod} \; 2
\ee
for any $F \in \Omega^{d/2}_\mathbb{Z}(X)$. This follows from the corresponding property of $\nu_{d/2}$ on $H^{d/2}(X;\mathbb{Z}_2)$ and the compatibility of the wedge and cup product pairings. A direct consequence of this fact is 

\begin{proposition} 
\label{PropIntAnSD}
Let $W$ be a closed 8-manifold and $\lambda$ be a form lift of the Wu class of degree 4. The expression 
\be
\frac{1}{8} \int_W (L(TW) - \lambda^2)
\ee
takes integer values, where $L(TW)$ is the Hirzebruch L-genus of $TW$.
\end{proposition}
\begin{proof}
The norm of any characteristic element of a unimodular lattice is equal to the signature modulo 8. (This is a special case of Theorem 2.9 of \cite{Brumfiel1973}, valid for any lattice.) The proposition then follows from the fact that the integral of the L-genus over the manifold yields the signature.
\end{proof}

{
\small
%\bibliographystyle{../../Bibliographie/BibStyle/utphys}
%\bibliography{../../Bibliographie/References}

\providecommand{\href}[2]{#2}\begingroup\raggedright\endgroup

}

\end{document}